\newif\ifjournal
\journalfalse

\ifjournal
\documentclass{birkjour}
\else
\documentclass[reqno,12pt]{amsart}
\usepackage[pdftex]{graphicx}
\usepackage[pdftex]{hyperref}
\hypersetup{
    unicode=false,          
    pdftoolbar=true,        
    pdfmenubar=true,        
    pdffitwindow=false,     
    pdfstartview={FitH},    
    pdfnewwindow=true,      
    colorlinks=true,        
    linkcolor=red,          
    citecolor=blue,         
    filecolor=magenta,      
    urlcolor=cyan           
}
\fi

\ifjournal
\usepackage{caption}
\fi

\usepackage{times,fullpage}
\usepackage{color} 
\usepackage{makecell}
\usepackage{enumitem}

\usepackage{amsmath}
\usepackage{amsthm}
\usepackage{amssymb}
\usepackage{amsbsy}
\usepackage{amscd}
\usepackage{verbatim}

\newtheorem{theorem}{Theorem}[section]

\numberwithin{equation}{section}  

\newcommand{\calM}{\mathcal{M}}

\newcommand{\abs}[1]{\left|#1\right|}
\DeclareMathOperator{\ck}{\mathcal{L}}

\DeclareMathOperator{\Lap}{\Delta}
\renewcommand{\div}{\mathop{\rm div}\nolimits}
\newcommand{\dtt}{\frac{d\tau}{\tau}}

\newcommand{\dxx}{\frac{d\xi}{\xi}}

\newcommand{\R}{\mathbb{R}}

\ifjournal\else
\fi

\newlength\plotheight
\ifjournal
\else
\fi

\newcounter{mnotecount}[section]
\let\oldmarginpar\marginpar
\renewcommand\marginpar[1]{\-\oldmarginpar[\raggedleft\footnotesize #1]%
{\raggedright\footnotesize #1}}

\begin{document}

\title[Numerical Bifurcation Analysis of the Conformal Method]
      {Numerical Bifurcation Analysis of the Conformal Method}

\author[J. Dilts]{James Dilts}
\email{jdilts@ucsd.edu}

\author[M. Holst]{Michael Holst}
\email{mholst@ucsd.edu}

\author[T. Kozareva]{Tamara Kozareva}
\email{txk171630@utdallas.edu}

\author[D. Maxwell]{David Maxwell}
\email{damaxwell@alaska.edu}

\address{Department of Mathematics\\
         University of California, San Diego\\ 
         La Jolla CA 92093}
\address{Department of Mathematics\\
         University of Texas, Dallas \\
         Dallas, TX 75080}
\address{Department of Mathematics\\
         University of Alaska, Fairbanks \\ 
         Fairbanks, AK 99775}
\thanks{JD was supported in part by NSF DMS/RTG Award 1345013 and DMS/FRG Award 1262982.}
\thanks{MH was supported in part by NSF DMS/FRG Award 1262982 and NSF DMS/CM Award 1620366.}
\thanks{TK was supported in part by NSF DMS/RTG Award 1262982.}
\thanks{DM was supported in part by NSF DMS/FRG Award 1263544.}

\date{\today}

\begin{abstract}
The conformal formulation of the Einstein constraint equations has been 
studied intensively since the modern version of the conformal method was 
first published in the early 1970s.
Proofs of existence and uniqueness of solutions were limited to the 
constant mean curvature (CMC) case through the early 90s,
with analogous results for the near-CMC case beginning to appear
thereafter.  In the last decade, there has been some limited
progress towards understanding the properties of the conformal
method for far-from-CMC solutions as well.  Although it was 
initially conceivable that that these far-from-CMC results
would lead to a solution theory for the non-CMC case that would mirror the 
good properties of the CMC and near-CMC cases,  examples of bifurcations and
of nonexistence of solutions have been since discovered.  Nevertheless, the general properties of the conformal method for far-from-CMC data remain
unknown. In this article we apply analytic and numerical continuation techniques 
to the study of the conformal method, in an attempt to give some insight 
into what the solution behavior is in the far-from-CMC case in various
scenarios.
\end{abstract}

\vspace*{-0.5cm}
\maketitle
\tableofcontents

\section{Introduction}
   \label{sec:intro}

In general relativity spacetime is described by by
a Lorentzian manifold $(\calM,g)$,
that is, a four-dimensional differentiable manifold $\calM$ endowed with 
a non-degenerate, symmetric rank $(0,2)$ tensor field 
$\mathbf g$ on $\calM$ whose signature is 
$(-1,1,1,1)$. 
The space-time $(\calM,g)$ is required to satisfy the
 {\it Einstein field equations}, 
\begin{equation}
\mathrm{Ric}_{\mathbf g} - \frac{1}{2} R_{\mathbf g}
       = \frac{8 \pi G}{c^4} T \, , \label{eq:einstein}
\end{equation}
where  $\mathrm{Ric}_{\mathbf g}$ 
is the Ricci curvature tensor, $R_{\mathbf g}$ its scalar 
($R:=\mathrm{Ric}_{ab}g^{ab}$), and $T$ is the stress 
energy-momentum tensor of any matter fields present.
Once a time function has been chosen, space-time is foliated by
space-like constant-time hypersurfaces $\Sigma_t$ and evolution 
and constraint equations are obtained by considering the projections 
of the field equations~\eqref{eq:einstein} in directions tangent and 
orthogonal to the space-like hypersurfaces. The evolution equations 
can be cast as a first-order system for the first and second 
fundamental forms associated with the time slices, namely the 
three-metric $\hat g$ and extrinsic curvature $\hat k$.
With $\hat g$ and $\hat k$ symmetric tensors, this represents 12
equations for the 12 components of $\hat g$ and $\hat k$, 
with the equations being first-order in time and second-order in space.

The four constraint equations on the 12 degrees of freedom are
\begin{eqnarray}
R_g + \tau^2 - |\hat k|_g^2 &=& 0,
\label{Eq:HamConstr}\\
\div \hat k - d\tau & = & 0
\label{Eq:MomConstr}
\end{eqnarray}
with $R_g$ the scalar curvature of $g$
and $\tau = \hat g^{ij} \hat k_{ij}$ the trace of the extrinsic curvature. 
These equations are direct consequences of the Gauss-Codazzi-Mainardi
conditions which are required for an $3$-manifold to arise as a 
submanifold of a $4$-manifold.
If matter and/or energy sources are present, then the 12 evolution equations
and the four constraint equations~\eqref{Eq:HamConstr}--\eqref{Eq:MomConstr}
contain additional terms.
The constraint equations are obviously underdetermined as a stand-alone
system of equations for the initial data, in that they fix only some
part of the 12 degrees of freedom.
One must therefore make a choice of which parts of the initial data one wishes
to fix, and which parts are to be determined by the constraint equations
~\eqref{Eq:HamConstr}--\eqref{Eq:MomConstr}.
The \emph{conformal method}, described in the next section, is an approach
to parameterizing the initial data so that the constraint equations for
the remaining degrees of freedom can potentially be uniquely solved.
It provides an effective parameterization of the constant-mean-curvature
(CMC) solutions of the constraint equations, and is generally effective
for near-constant mean curvatures as well.  There has been recent 
progress in determining its properties in the far-from-CMC setting, and
what little we know indicates the situation is somewhat complex.  The
aim of this paper is to bring numerical methods, and numerical bifurcation theory specifically, to yield further insight into what can be expected
for the conformal method when applied to far-from-CMC initial data.

\subsection{The Conformal Method}
   \label{sec:conformal-method}

The conformal method was proposed by Lichnerowicz in 1944~\cite{L44},
and then substantially generalized in the 1970s by York~\cite{Y73},
among other authors.
The method is based on a splitting of the initial data $\hat{g}$
(a Riemannian metric on a space-like hypersurface $\Sigma_t$) 
and $\hat{k}$ (the extrinsic curvature of the hypersurface 
$\Sigma_t$) into eight freely specifiable pieces, with four remaining 
pieces to be determined by solving the four constraint equations.

The pieces of the initial data 
that are specified as part of the method are called the \emph{seed data} 
and are comprised of a spatial background metric $g$ on $\Sigma_t$, 
defined up to multiplication by a conformal factor (five free functions), 
a positive function $N$ (a so-called densitized lapse), 
a function $\tau$, and a transverse, traceless (TT) tensor $\sigma_{ij}$ 
(effectively two free functions, as it is symmetric, trace-free 
and divergence free).
The two remaining pieces of the initial data to be determined by 
the constraints are a scalar conformal factor $\varphi>0$ and 
a vector potential $W$.
The full spatial metric $\hat{g}$ and the extrinsic curvature 
$\hat{k}$ are then recovered from $\varphi$, $W$, and the 
eight specified functions from the expressions:
$
\hat{g} = \varphi^4 g,
$
and
$
\hat{k} = \varphi^{-2}\left[\sigma + \frac{1}{2N}(ck W)\right]
    + \frac{1}{3} \varphi^{4} \tau g.
$
This transformation has been engineered so that the constraints~
\eqref{Eq:HamConstr}--\eqref{Eq:MomConstr} reduce to coupled PDEs
for $\varphi$ and $W$ with standard elliptic operators as their principle parts; in three spatial dimensions the equations are
\begin{align}
- 8 \Delta \varphi 
    + R \varphi
    + \frac{2}{3} \tau^2 \varphi^5
    - \left|\sigma+\frac{1}{2N}\mathcal{L}W\right|^2
      \varphi^{-7}
&=0,
\label{eqn:ham_conf}
\\
- \div\left[\frac{1}{2N} \ck W\right] + \frac{2}{3} \varphi^6 d\,\tau 
&=0.
\label{eqn:mom_conf}
\end{align}
Here, $\Delta$ is the Laplace-Beltrami operator with respect to the 
background metric $g$, $\mathcal{L}$ denotes the 
\emph{conformal Killing operator}
$(\mathcal{L}W)_{ij} = \nabla_i W_j + \nabla_j W_i - \frac{2}{3} (\nabla_k W^k) g_{ij}$, 
and $\tau= \hat{k}_{ij}\hat{g}^{ij}$ is again the trace of the extrinsic curvature. 
We note that the densitized lapse $N$ is more commonly associated with the
conformal thin-sandwich method\cite{Y99}, but the equivalence of
that method with the
standard conformal method was demonstrated in \cite{Maxwell:2014a}.
In particular, the conformal method represents a family of parameterizations
of the constraint equations within a given conformal class of metric, 
one for each choice of $N$ (or one for each choice of metric representing the conformal class). A detailed overview of the conformal method, 
and its variations, may be found 
in the 2004 survey \cite{BI04}. 

When $\tau$ is constant (i.e. when the Cauchy surface $\Sigma_t$ has constant mean curvature), then the term in equation~\eqref{eqn:mom_conf} 
involving $d\, \tau$ vanishes, and the two equations decouple.
The only solutions of~\eqref{eqn:mom_conf} have $\ck W=0$,
and it remains only to solve the Lichnerowicz equation
~\eqref{eqn:ham_conf} with $\ck W=0$; a similar decoupling occurs for non-vacuum seed data data.
Initial work starting with \cite{MuYo73} focused on the CMC
case of the conformal method, and a full description of the 
parameterization on compact manifolds 
was achieved in \cite{JI95}.  The theory
depends on on the Yamabe invariant $Y(g)$ of the seed metric\footnote{Recall
that $Y(g)>0$ if and only if $g$ has a conformally related 
metric with positive scalar curvature, and similarly for $Y(g)=0$
and $Y(g)<0$.}, and 
is summarized in Table~\ref{table:cmcSolv}.
\begin{center}
\begin{table}[h]
\caption{Constant mean curvature (CMC) solvability \cite{JI95}}
\centering
\begin{tabular}{|c|c|c|c|c|}
\hline
       & $\tau=0, \sigma \equiv 0$ & $\tau=0, \sigma \not\equiv 0$ &
         $\tau\neq 0, \sigma\equiv 0$ & $\tau\neq 0, \sigma\not\equiv 0$  \\ \hline
$Y(g)>0$ & None   & Unique & None  & Unique \\ \hline
$Y(g)=0$ & Unique up to homotopy & None  & None  & Unique \\ \hline
$Y(g)<0$ & None   & None  & Unique & Unique \\ \hline
\end{tabular}
\label{table:cmcSolv}
\end{table}
\end{center}
The CMC conformal method is also well understood in other 
asymptotic geometries (e.g, asymptotically Euclidean \cite{CBI00}\cite{dM05b}, asymptotically hyperbolic \cite{AC96}).  
It has been used in a number of applications, including results for open 
manifolds with interior ``black hole'' boundary models~\cite{sD04,dM05b}, results allowing for ``rough'' data~\cite{DMa06,yCB04},
and numerical relativity (e.g.~\cite{Cook91,CoTe99})

Investigations of near-CMC seed data began to appear in the mid-90s,
and we point to \cite{IM96}, \cite{ACI08} and \cite{IOM04}
which developed the near-CMC theory\footnote{The specific conditions characterizing near-CMC seed data depend on the context but all
involve control on $d\tau/\tau$.} on compact manifolds
summarized in Table~\ref{table:conjSolv}.
\begin{center}
\begin{table}[h]
\caption{Near-CMC Solvability  \cite{IM96,ACI08,IOM04}}
\centering
\begin{tabular}{|c|c|c|}
\hline
       & $\tau\not\equiv 0, \sigma\equiv 0$ & $\tau\not\equiv 0, \sigma\not\equiv 0$  \\ \hline
$Y(g)>0$ & {None}  & {Unique}\\ \hline
$Y(g)=0$ & None  & {Unique} \\ \hline
$Y(g)<0$ & Unique & Unique \\ \hline
\end{tabular}
\label{table:conjSolv}
\end{table}
\end{center}
The existence results in this table
require an additional hypothesis that the background metric not have
any conformal Killing fields, and it has been recently shown
\cite{HMM2017} that in some cases this hypothesis is necessary.
Conformal Killing fields form the kernel of the 
self-adjoint elliptic operator appearing
in equation \eqref{eqn:mom_conf}, and their presence interferes with
iterative approaches to obtaining solutions; nearly all theorems
concerning non-CMC seed data for the conformal method assume there are no
conformal Killing fields.  
Extensions of the near-CMC theory are available in other asymptotic
geometries (asymptotically Euclidean \cite{CBI00}, asymptotically
hyperbolic \cite{IP}), and it has been applied in
numerical relativity~\cite{Cook:2000vr,BeHo96,Hols2001a,Pfei04}.

\subsection{Far-from-CMC results for the conformal method}
\label{ssec:ffcmcresults}

In the last decade, a handful of results have appeared 
concerning the conformal method in the far-from CMC setting.
They provide the main context needed to understand our numerical 
experiments, and we summarize them in somewhat more detail 
in this section.

Building off of a strictly non-vacuum result \cite{HNT07b},
the following theorem provides existence, in vacuum, for arbitrary 
mean curvatures, so long as the background metric is Yamabe positive 
and the TT tensor is sufficiently small.
\begin{theorem}[\cite{dM09}]\label{thm:smallTT} Let $(M,g)$ be a compact Yamabe-positive
manifold with no conformal Killing fields.  Given arbitrary vacuum
seed data $(\tau,\sigma,N)$, if $\sigma\not\equiv 0$ 
is sufficiently small (with smallness depends on the choice of $\tau$), 
then there exists at least one solution of the conformally
parameterized constraint equations \eqref{eqn:ham_conf}-\eqref{eqn:mom_conf}.
\end{theorem}
Variations on this theorem have subsequently been demonstrated in other 
contexts (asymptotically Euclidean 
manifolds~\cite{DIMM14,HoMe14a,BeHo14a}, 
manifolds with asymptotically cylindrical or 
periodic ends~\cite{CM12,CMP12}, and other settings
~\cite{HoTs10a,HMT13a,Dilts:2013,HNT07b,BeHo14a}). 
Although  
Theorem \ref{thm:smallTT} is silent on the issue of uniqueness,
it is consistent with Table \ref{table:conjSolv}
extending generally for arbitrary mean curvatures, and there was
some optimism that this might be the case when
\cite{HNT07b,dM09} appeared.  Although it is evidently
a far-from CMC result, the alternative
perspective of \cite{GiNg14a} demonstrates that the solutions
found in Theorem \ref{thm:smallTT} can also be thought of as rescalings
of near-CMC solutions that are perturbations off of $\tau\equiv 0$
solutions, as allowed in Table \ref{table:cmcSolv} for Yamabe-positive
seed data.  

The following result is a consequence of a groundbreaking blowup
analysis for the conformal method.
\begin{theorem}[\cite{DGH10}]\label{thm:limiteq} Let $(M^n,g)$ be a compact manifold
without conformal Killing fields,
and let $(\sigma,\tau,N)$ be vacuum seed data on it with $\tau\ne 0$ having
constant sign.  If
there does not exist a solution of the conformally
parameterized constraint equations 
\eqref{eqn:ham_conf}-\eqref{eqn:mom_conf}, then 
there exists a solution of the \textit{limit equation}
\begin{equation}\label{eq:limitEqn}
  \div\left[ \frac{1}{2N} \ck W\right] = \alpha \sqrt{\frac{n-1}{n}} \left|\frac{1}{2N}\ck W\right| \dtt.
\end{equation}
for some $\alpha\in(0,1]$.
\end{theorem}
See also \cite{GS12,GICQUAUD2016175}, where
Theorem \ref{thm:limiteq} has been extended to other settings.
The main idea behind the proof of the theorem is that if one cannot
maintain $L^\infty$ control on approximate solutions $\phi$ of
the conformally parameterized constraint equations, then rescalings
of the approximations eventually lead to a solution
of the blowup profile \eqref{eq:limitEqn}.
One potential application of Theorem \ref{thm:limiteq} is to show
solutions exist by ruling out the possibility of solutions
of the limit equation \eqref{eq:limitEqn}, and indeed
\cite{DGH10} contains a near-CMC existence theorem with a large
perturbation constant based on this idea. 
Although Theorem \ref{thm:limiteq} has proved difficult to apply
in practice,
because of the challenge of working with 
equation \eqref{eq:limitEqn}, our numerical work suggests that
it plays a decisive role in analyzing system 
\eqref{eqn:ham_conf}-\eqref{eqn:mom_conf} for constant sign mean curvatures.

For mean curvatures that change sign, little is known in general aside
from Theorem \ref{thm:smallTT}.  However, \cite{M11} contains
an analysis of some very specific families of seed data
on the flat torus that includes the sign changing case.
Among the seed data considered there is a family 
of mean curvatures of the form
\[
\tau = 1 + a\xi
\]
where $\xi$ is a particular dicontinuous, piecewise constant function equal
to $\pm 1$ and TT tensors of the form $\eta\hat \sigma$ 
for a particular reference TT tensor $\hat \sigma$ and 
an arbitrary constant $\eta$.
\begin{theorem}[\cite{M11}]\label{thm:modelp}
For particular seed data on the flat torus $T^3$ of the form
\[
(\tau=1+a\xi,\sigma=\eta\hat \sigma, N),
\]
if $|a|>1$ (and hence $\tau$ changes sign), 
there is an $\eta_*>0$ depending on $a$ so that if 
$\eta>\eta_*$ there is no solution of 
system \eqref{eqn:ham_conf}-\eqref{eqn:mom_conf} with the
symmetry of the data, but if $0<\eta<\eta_*$ then there 
are at least two solutions.
\end{theorem}
This was the first theorem to demonstrate the existence
of multiple solutions of the vacuum 
conformal method in the far-from-CMC setting. 
Although in involves Yamabe-null examples, it cast
doubt on the possibility that Theorem \ref{thm:smallTT}
concerning Yamabe-positive seed data
could be extended to include a uniqueness statement or that
its small-TT tensor requirement could be dropped.  The
follow-up study \cite{Maxwell:2014b} contains
additional results on related families of far-from-CMC data.

Using ideas from \cite{DGH10}, Nguyen \cite{Nguyen:2015}
recently showed conclusively that the restrictions of 
Theorem \ref{thm:smallTT} are essential.
\begin{theorem}[\cite{Nguyen:2015}] \label{thm:Nguyen}
Let $(M,g)$ be a compact Yamabe-positive manifold $(M,g)$ 
with no conformal Killing fields.  Consider a family of seed data
$(\tau = \xi^a, \mu\sigma,N)$  with $a>0$, $\mu\in \R$, where $\xi$
is a fixed positive function.  Assume additionally $\xi$ satisfies
\begin{equation}\label{eq:almostckf}
\left|\mathcal{L}\left(\dxx\right)\right| \leq c \left|\dxx\right|^2
\end{equation}
for some $c>0$, and that $\sigma$ is supported away from the critical
points of $\tau$. Then if $a$ is sufficiently large, 
and if $|\mu|$ is larger than a 
threshold depending on $a$,
the conformally parameterized constraint equations 
\eqref{eqn:ham_conf}-\eqref{eqn:mom_conf} do not admit a solution.
For the same $a$, there is a sequence $\mu_k\to 0$
such that there are at least two solutions to these equations along the sequence, 
and such that there is a solution with $\mu =0$. 

The set of seed data satisfying these conditions is nonempty.
\end{theorem}
The restrictions on the seed data in Theorem
\ref{thm:Nguyen} are quite severe, but the result
is remarkable nevertheless.  In particular, the existence
of solutions at $\sigma\equiv0$  for Yamabe-positive data
was a surprise.  In addition to proving Theorem \ref{thm:Nguyen}, 
\cite{Nguyen:2015} gives insight into the role
of the deficiency parameter $\alpha$ in the limit equation 
\eqref{eq:limitEqn} and these unexpected $\sigma\equiv 0$
solutions. On a Yamabe-positive manifold, if there 
does not exist a solution for given seed data, and if there
does not exist a solution of the limit equation with $\alpha=1$,
then there is, in fact, a solution for the same seed data but with
$\sigma\equiv 0$.

Very recently, Nguyen obtained the following extension of Theorem \ref{thm:Nguyen}.
\begin{theorem}[\cite{tcN2018}]\label{thm:ngn2}  
Let $(M,g)$ be a compact Yamabe-positive manifold $(M,g)$ 
with no conformal Killing fields.  Consider a family of seed data
$(\tau = \xi^a, \mu\sigma,N)$  with $a>0$, $\mu\in \R$, and where $\xi$
is a fixed positive function satisfying inequality \eqref{eq:almostckf}.
Then, for each $a$ sufficiently large, there is a $m_0$ depending on 
$a$ such that if $0<\mu<m_0$ there are at least two solutions 
of the conformally parameterized constraint equations, and
when $\mu=0$ there is at least one.
\end{theorem}
While Theorem \ref{thm:ngn2} significantly relaxes many of the hypotheses
of Theorem \ref{thm:Nguyen} and strengthens some of its conclusions, 
it makes no claims concerning non-existence for large $\mu$, 
a point we will revisit in our numerical experiments. 
Inequality \eqref{eq:almostckf} needed for Theorems \ref{thm:Nguyen}
and \ref{thm:ngn2} is non-generic, and our numerical work gives
insight about the extent to which Theorem \ref{thm:ngn2} holds more
generally.

In summary, we have the following.
\begin{itemize}
    \item On Yamabe-positive manifolds, arbitrary mean curvatures
    can be used for seed data, so long as the TT tensor is small enough.
    For certain families of Yamabe-positive seed data, there are 
    multiple solutions for small TT tensors. Moreover, for these families
    of seed data, there exist
    solutions at $\sigma\equiv0$, something ruled out in
    the near-CMC case (Table \ref{table:conjSolv}).  Additionally,
    there are some cases where there are no solutions for large
    TT tensors.
    \item On a particular Yamabe-null manifold, with particular 
    sign-changing seed data, we have multiple solutions for small TT tensors,
    and nonexistence (within the symmetry class) for large TT tensors.
    \item Nothing specific is known for Yamabe-negative seed data.
    \item The limit equation criterion holds for all Yamabe classes,
    but the question of the existence of solutions of the limit
    equation is essentially open.
\end{itemize}
These limited results provide
our motivation to look at analytic 
bifurcation theory and closely related numerical continuation methods to try 
to gain intuition for what can be expected more generally
 from the conformal method in the far-from CMC regime.

\section{Tools from Bifurcation Analysis}

The unexpectedly complex behavior
of solutions to the conformal method equations in the far-from-CMC regime
leads one to the language and technical tools of 
\emph{analytic bifurcation theory}.
This area of nonlinear analysis is the study of the branching of solutions 
of nonlinear problems with respect to the parameters.
\emph{Numerical continuation} (or \emph{numerical homotopy methods})
is a related area and refers to a collection 
of practical numerical methods for computing solution branches of nonlinear 
problems through critical points such as folds and bifurcations.

\subsection{Analytic Bifurcation Theory}
   \label{sec:bifur-analytic}

To explain the main ideas that are relevant here,
consider again the PDE representation of the 
conformal method~\eqref{eqn:ham_conf}--\eqref{eqn:mom_conf},
but written more simply 
as the abstract nonlinear problem: Find $u \in X$ such that
\begin{align}
F(u,\lambda) &= 0, 
\label{eq:fulam}
\end{align}
where $F\colon X \times Z \to Y$
for suitably chosen Banach spaces $X$, $Y$, and $Z$,
and where $\lambda \in Z$ represents the parameters of interest
that are moved through the parameter space $Z$.
In the setting of the conformal method and its variations, 
one can consider various parameterizations.
Once a parameterization $\lambda$ is chosen,
one is interested in the \emph{local behavior} of the solution curve 
$u(\lambda)$ in a neighborhood of a known solution $u_0(\lambda_0)$.
The techniques of both analytic and numerical bifurcation analysis
rely on the \emph{Implicit Function Theorem (IFT)} as the basic tool for
doing this exploration.

\emph{
Given $F:X\times Z\rightarrow Y$, 
where $X$, $Y$, and $Z$ are Banach spaces,
if $F(u_0,\lambda_0)=0$,
if $F$ and $F_u$ (the Frech\'{e}t derivative of $F$) are continuous on some 
region $U\times V \subset X \times \mathbb{R}$ containing $(u_0,\lambda_0)$,
and if $F_u(u_0,\lambda_0)$ is nonsingular with a bounded inverse, 
then there is a unique branch of solutions $(u(\lambda),\lambda))$ 
to $F(u(\lambda),\lambda)=0$ for $\lambda\in V$.
Moreover, $u(\lambda)$ is continuous with respect to $\lambda$ in $V$.
}

The IFT effectively states that if the linearization $F_u$ of the nonlinear 
operator operator $F$ is nonsingular at the point $[u_0,\lambda_0]$,
then there is a unique solution $u(\lambda)$ for each $\lambda$ in a ball
around $\lambda_0$. More details about this theorem 
and its proof can be found in \cite{HK04,EZ86,CH82}.

If $F_u$ is singular, however, the proof of the IFT fails, suggesting the 
possibility of two or more $u(\lambda)$ branches, or no solutions, 
for some $\lambda$ in every neighborhood around $\lambda_0$.
The form of the branching depends on the structure of the subspaces 
associated with the linear maps 
$F_u(u_0,\lambda_0)$ and $F_\lambda(u_0,\lambda_0)$.
In the case of a ``fold'', there is a one-dimensional path through 
$[u_0,\lambda_0]$; in the case of a simple (or more general) singular point, 
there is the possibility of branch-switching, with two (or more) branches 
of solutions crossing through $[u_0,\lambda_0]$.

One of the central tools in analytic bifurcation theory is 
\emph{Lyapunov-Schmidt Reduction}~\cite{Zeid91a}.
To explain, assume
$F_u(u_0,\lambda_0)$ is a Fredholm operator of index $k$,
and that $\text{dim}(\mathcal{N}(F_u(u_0,\lambda_0)))=n$.
Define now projection operators $P:X\rightarrow X$ and $Q:X \rightarrow X$ 
with $P(X)=\mathcal{N}(F_u(u_0,\lambda_0))$ 
and $(I-Q)(Y)=\mathcal{R}(F_u(u_0,\lambda_0))$. 
Equation~\eqref{eq:fulam} is equivalent to the pair
\begin{align}
(I-Q)F(y+z,\lambda)&=0, \label{eq:LSa} \\
QF(y+z,\lambda)&=0, \label{eq:LSb}
\end{align}
with $y=(I-P)u$ and $z=Pu$. 
Equation~\eqref{eq:LSa} satisfies the assumptions of the IFT, and so one 
obtains a unique solution branch $y(z,\lambda)$, then substitutes
into~\eqref{eq:LSb} to obtain the \emph{branching equation}:
\begin{equation}\label{eq:branch}
QF(y(z,\lambda)+z,\lambda)=0.
\end{equation}
One then solves for $z(\lambda)$ to get the branch 
$u=y(z(\lambda),\lambda)+z(\lambda)$.
In practice, one solves \eqref{eq:LSa}--\eqref{eq:LSb} by expanding the 
operators in bases of $\mathcal{N}(F_u(u_0,\lambda_0))$ and 
$\mathcal{N}(F_u(u_0,\lambda_0)^*)$. 
A more detailed description of the application of 
Lyapunov-Schmidt Reduction to variations of the conformal
method may be found in~\cite{HoMe12a,ChGi15}.
One of our goals here is to apply the reduction technique to the 
far-from-CMC parameterizations that were described earlier. 
More information on this decomposition can be found in \cite{HK04}.

\subsection{Numerical Bifurcation Analysis}
   \label{sec:bifur-numerical}

To apply Lyapunov-Schmidt reduction analytically, one needs detailed 
information about the null and range spaces of the linearization operators 
$F_u$ and $F_u^*$, and therefore the technique is usually limited 
to model situations.
However, by \emph{discretizing} problem~\eqref{eq:fulam}, it becomes
tractable to explicitly compute the information one needs
for a finite-dimensional approximation of~\eqref{eq:fulam}.
The problem retains the structure of~\eqref{eq:fulam}, but the discretized
problem now involves finite-dimensional spaces
$X=Y=\mathbb{R}^n$ and $Z=\mathbb{R}^m$, where $n$ is the resolution of
the discretization (e.g. number of finite element basis functions),
and $m$ is the number of parameters.
One now numerically computes bases explicitly for the range and null
spaces of what are now \emph{matrix} operators $F_u$ and $F_{\lambda}$.
Moreover, a \emph{numerical continuation algorithm} can be designed
around a predictor-corrector strategy: one increments the parameter
$\lambda_0 \to \lambda_1$ as part of a prediction step, followed by 
the use of Newton's method to solve~\eqref{eq:fulam} to correct the solution 
$u(\lambda_0) \to u(\lambda_1)$.
Where the fold or higher-order singularity on the branch is
encountered, the linearization $F_u$ becomes singular, 
leading to failure of Newton's method at the correction step.
To remedy this,
one adds a \emph{normalization equation} $N(u(\lambda(s)),\lambda(s),s)=0$
that allows the larger coupled system involving 
$F$ and $N$ to again be solvable.
One of the standard normalization techniques is known as 
\emph{pseudo-arclength continuation}, based on parameterizing~$\lambda(s)$ 
by arclength~$s$.
These numerical techniques are well-studied~\cite{Kell87} and there are
well-established software packages that implement these techniques,
such as AUTO~\cite{AUTO-mainreference}.  We use AUTO
as our primary tool in our numerical analysis of the conformal method
for far-from-CMC seed data.

\section{Numerical Results}
   \label{sec:numerical-results}

The AUTO software package applies numerical bifurcation analysis
to systems of ordinary differential equations.  Thus, to apply it
to the conformal method, we require seed data with sufficient 
symmetry so that the conformally-parameterized constraint equations
\eqref{eqn:ham_conf}-\eqref{eqn:mom_conf} reduce to ODEs.  Once this
is done, the parameter space can be explored via homotopies starting
with CMC solutions.  The following subsections describe a number of
concrete datasets where we have done so and report on folds
and the number of solutions found as the mean curvature is made
increasingly far-from-CMC and as the size of the TT tensor is varied.
Some care is needed in interpreting our results.  If we find, e.g.,
two solutions for a given seed data set, it does not imply that there
are only two. Rather, with homotopies we chose 
starting from CMC data, we were only able to find two.  There may be more that we did not find because we did not, or were not able, to explore the parameter space more broadly.  A similar caveat applies when we find no solutions; there may be solutions that we did not find along our homotopies.  Additionally, because we seek solutions having the same symmmetry as 
the seed data (in practice, these are solutions depending 
on only one coordinate of the underlying 3-manifold), we cannot
rule out the possibility of additional solutions that break this symmetry.
Finally, because of the high symmetry of our seed data, our
metrics always admit conformal Killing fields, and hence violate 
a key technical hypothesis of most theorems concering the conformal
method with non-CMC seed data.

\subsection{Sign-changing mean curvature}
\label{subsec:sign-change}

In this section we examine properties of the conformal method
when the far-from-CMC regime is reached via a sign changing mean 
curvature. The conformal seed data has the following form:
\begin{itemize}
    \item The manifold is $S^1\times M_2$ where
    $M_2$ is one of $S^2$, $T^2$ or a compact quotient $H^2$
    of hyperbolic space.  We use $s$ for the unit speed parameter
    along $S^1$.
    \item The mean curvature is
    \[
    \tau = 1 + a\cos(s).
    \]
    So $a=0$ is the CMC case, and $\tau$ is sign-changing whenever
    $|a|>1$.
    \item We work with two different classes of TT tensors.
\begin{enumerate}
    \item On a product $(M_1^{n_1},g_1)\times (M_2^{n_2},g_2)$,
the tensor
\begin{equation}\label{eq:barsigma}
\bar \sigma = n_1 g_2 - n_2 g_1
\end{equation}
is easily seen to be transverse-traceless.  For many
of our experiments we use a TT tensor of the form $\mu\overline\sigma$,
where $\mu$ is a constant.

\item Additionally, on $S^1\times T^2$ and on $S^1\times S^2$ we 
can find a TT tensor $\hat \sigma$ with constant (nonzero) norm
that is pointwise orthogonal to $\ck W$ for $W=w(s)ds$.
These are easy enough to find on 
$T^2$ and Appendix \ref{apdx:TT} 
describes a suitable construction on $S^2$.
Some of our experiments make use of TT tensors
of the form $\sigma = \eta \hat \sigma$ where $\eta$ is a constant.  
\end{enumerate}

\item For simplicity of exposition, we use a lapse density $N=1/2$.
We conducted experiments with other choices for the lapse density
but did not see qualitatively different phenomena.
\end{itemize}
\begin{figure}
\includegraphics[height=\plotheight]{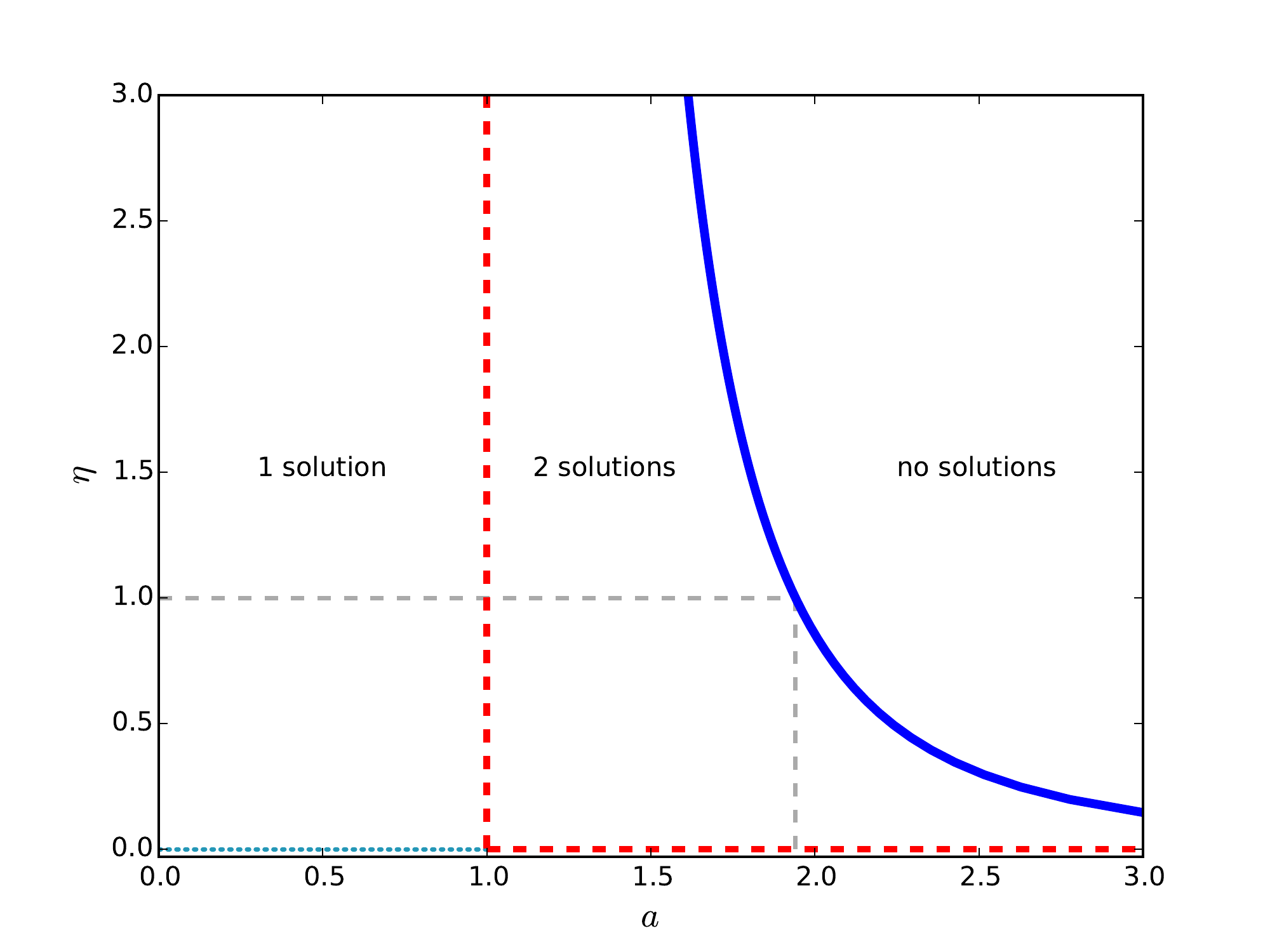}
\caption{\label{fig:S1-Yzero-mult}Multiplicity of solutions found on $S^1\times T^2$. Seed data:
$\tau=1+a\cos(s)$ and $\sigma = \eta \hat \sigma$.
The blue line is a computed fold, whereas the red dashed lines
indicate locations where blowup is inferred. The
blue dotted line indicates a zero-volume solution, which
should be discounted.  Solutions along the gray dashed lines 
are discussed in Figure \ref{fig:S1-Yzero-folds}.}
\includegraphics[height=\plotheight]{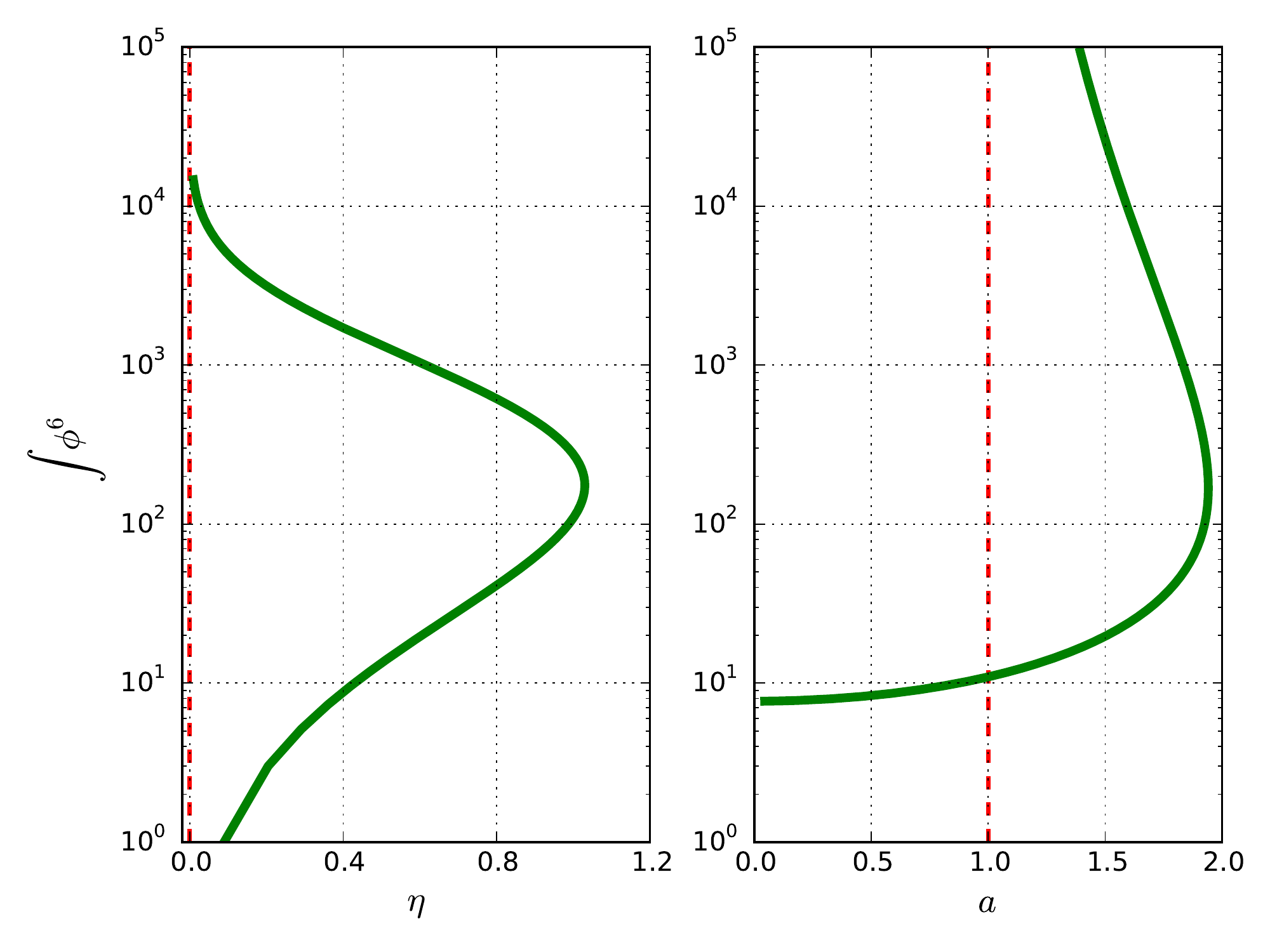}
\caption{\label{fig:S1-Yzero-folds}Volume of solutions on $S^1\times T^2$ 
as the size $\eta$ of the TT tensor (left-hand side)
and as the mean curvature $\tau=1+a\cos(s)$ (right-hand side) 
are varied.
The left-hand graph corresponds with the vertical gray
dashed line of Figure \ref{fig:S1-Yzero-mult} at $a=1.94$,
and  right-hand graph corresponds with the horizontal gray
dashed line at $\eta=1$.}
\end{figure}

For solutions $\phi=\phi(s)$ and $W=w(s)\partial_s$ of
the conformally-parameterized constraint equations 
\eqref{eqn:ham_conf}-\eqref{eqn:mom_conf} 
having the same symmetry as our seed data, 
the constraint equations reduce to the coupled ODEs
\begin{align}
\label{eqn:S1system}
  -8\varphi'' + R \varphi + \frac{2}{3} \tau^2 \varphi^5 &= 
                               \frac23 (\mu + 2w')^2\varphi^{-7} + \eta^2 \varphi^{-7},\\
  2w'' &= \tau' \varphi^6,
\end{align} where $R$ is the constant scalar curvature of the product manifold, and where we set $\eta=0$ if $R<0$.

\begin{figure}
\includegraphics[height=\plotheight]{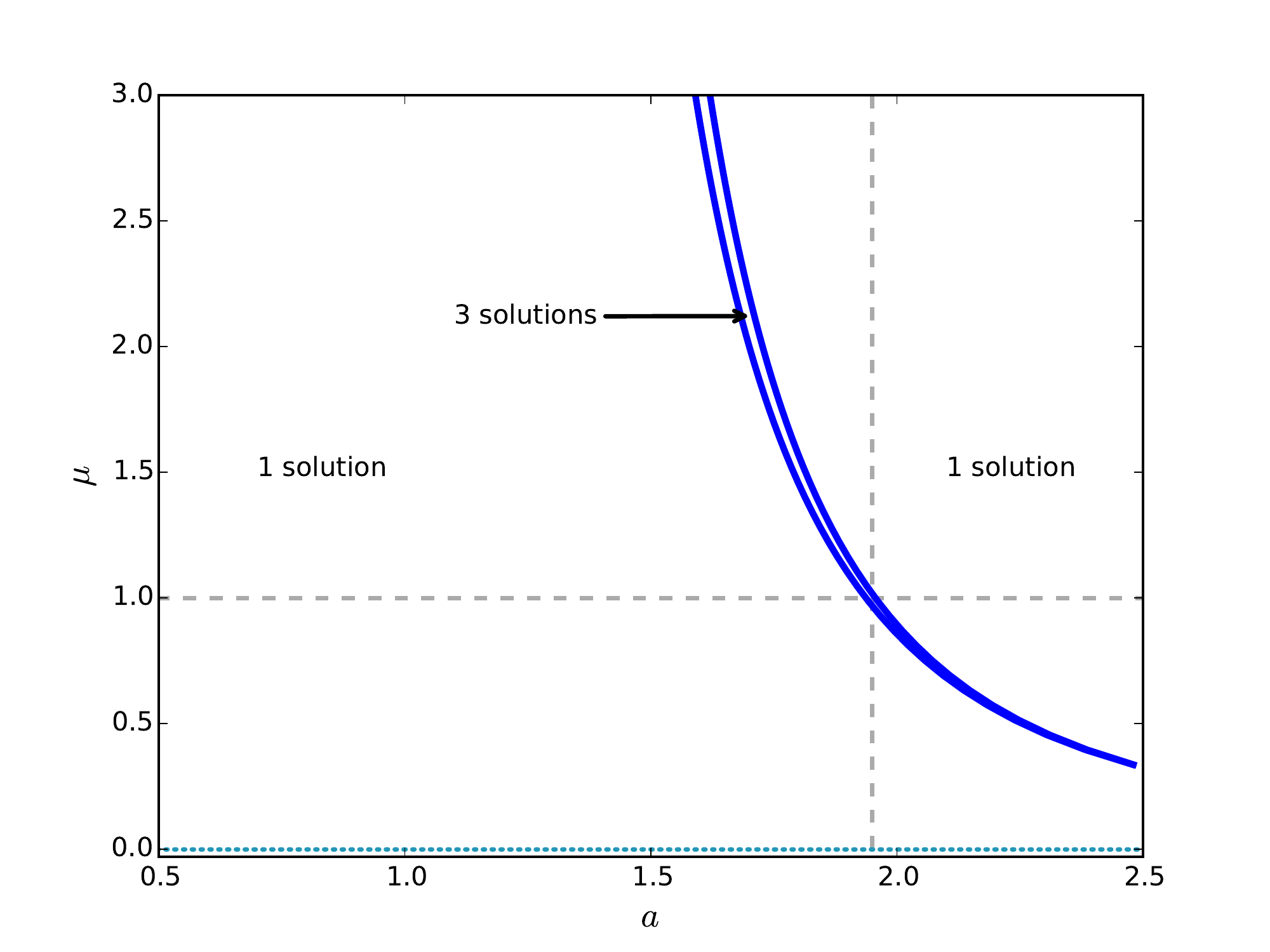}
\caption{\label{fig:S1-Ypos-mult}
Multiplicity of solutions found on an $S^1\times S^2$ 
with $R=0.001$. Here
$\tau=1+a\cos(s)$ and $\sigma = \eta \hat \sigma$.
The blue line is a computed fold. On the blue dotted line
at $\mu=0$  (i.e., $\sigma\equiv 0$) 
there is no solution; the volume has shrunk to zero.
Solutions along the gray dashed lines 
are discussed in Figure \ref{fig:S1-Ypos-folds}.}
\includegraphics[height=\plotheight]{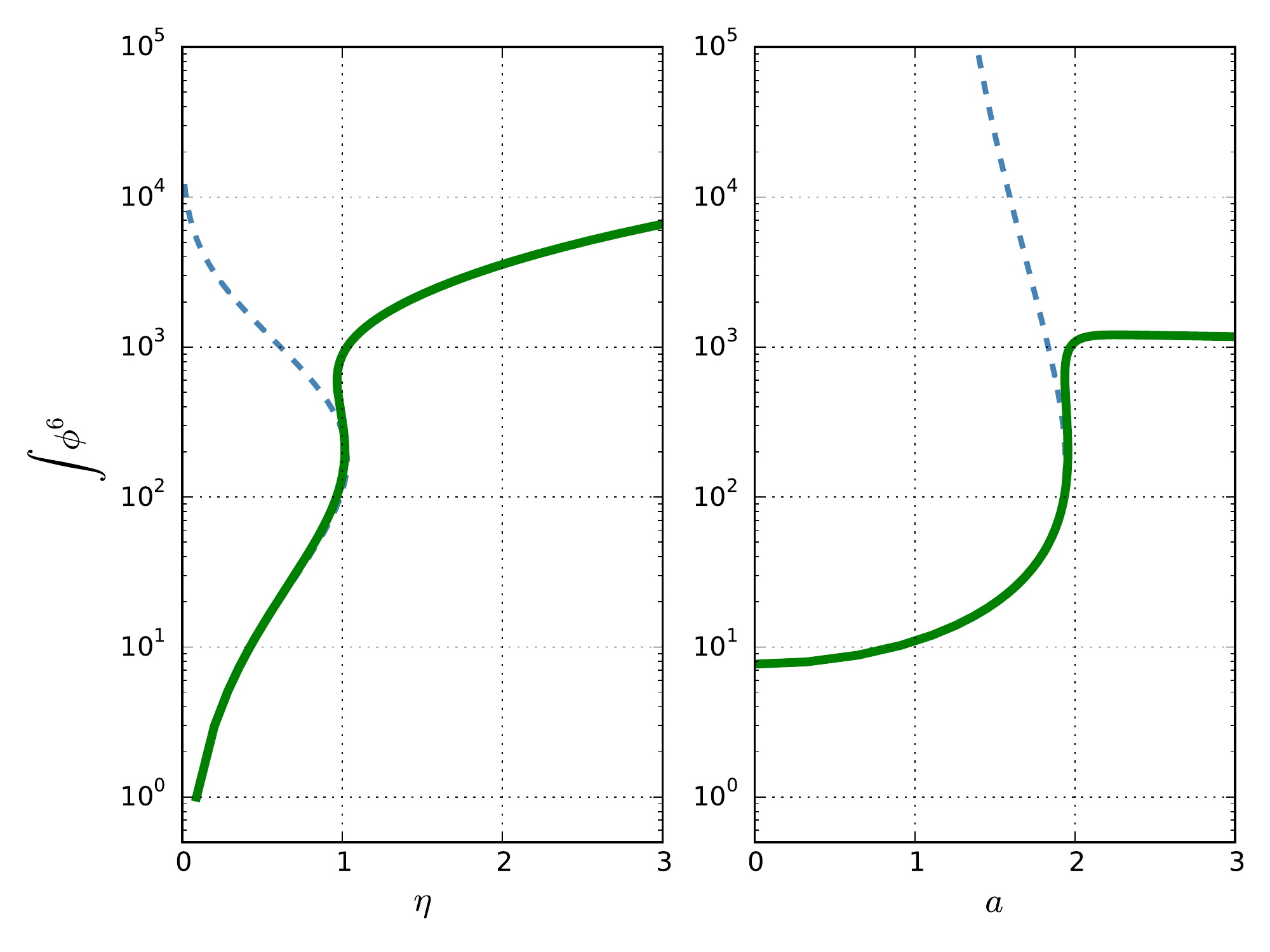}
\caption{\label{fig:S1-Ypos-folds}
Volume of solutions on an $S^1\times S^2$ 
with $R=0.001$
as the size of the TT tensor and 
as the mean curvature are varied.
The left and right-hand graphs correspond
with the vertical and horizontal gray
dashed lines of Figure \ref{fig:S1-Ypos-mult} respectively.
The dashed lines are the corresponding volumes
on $S^1\times T^2$ from Figure \ref{fig:S1-Yzero-folds} for comparison.}
\end{figure}

\begin{figure}
\includegraphics[height=\plotheight]{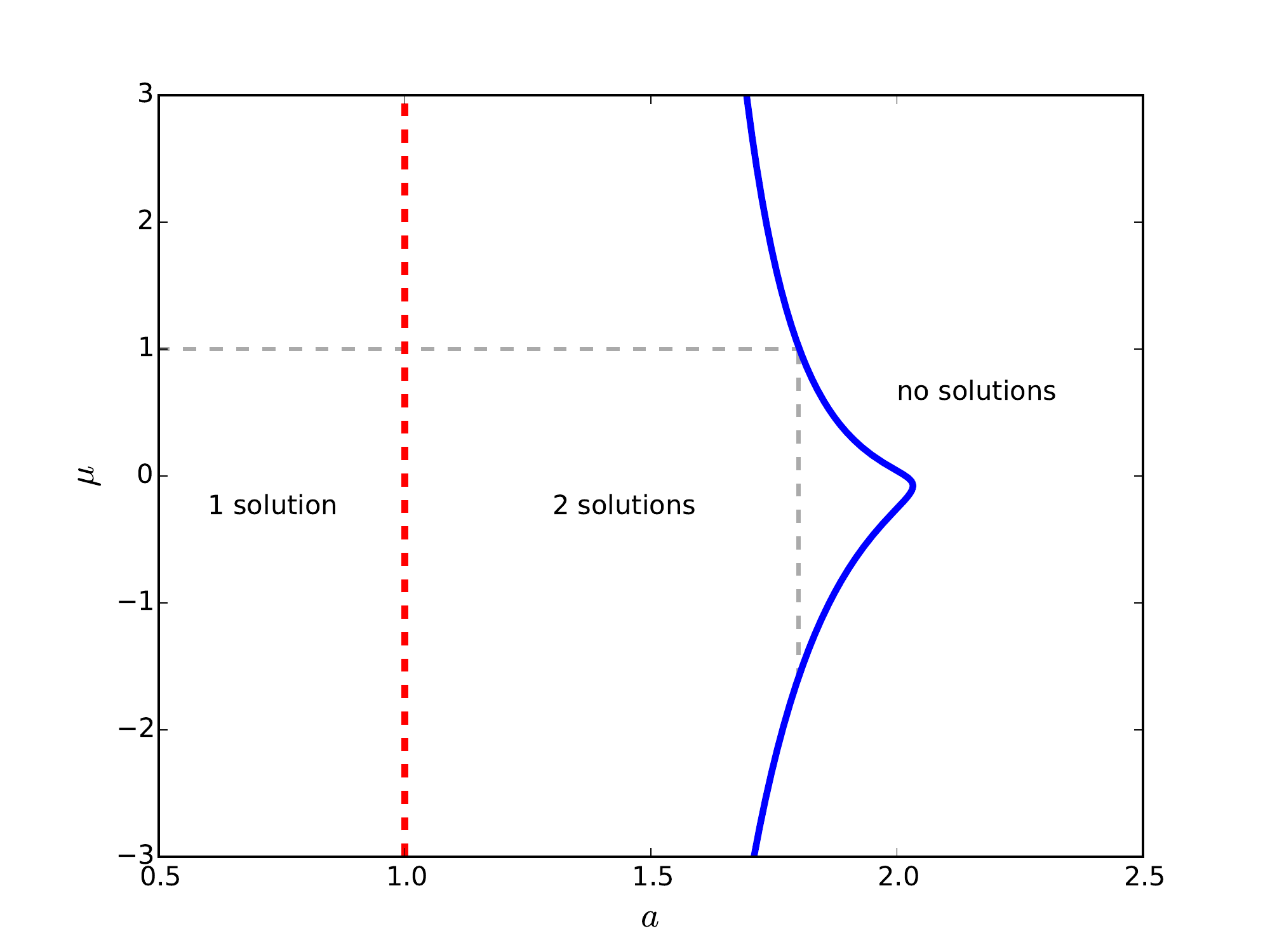}
\caption{Multiplicity of solutions found on an $S^1\times H^2$ 
with negative scalar curvature $R=-0.1$. Seed data:
$\tau=1+a\cos(s)$ and $\sigma = \mu \overline \sigma$.
The blue line is a computed fold, whereas the red dashed line
indicate locations where blowup is inferred. Solutions along 
the gray dashed lines 
are illustrated in Figure \ref{fig:S1-Yneg-folds}.}
\label{fig:S1-Yneg-mult}
\includegraphics[height=\plotheight]{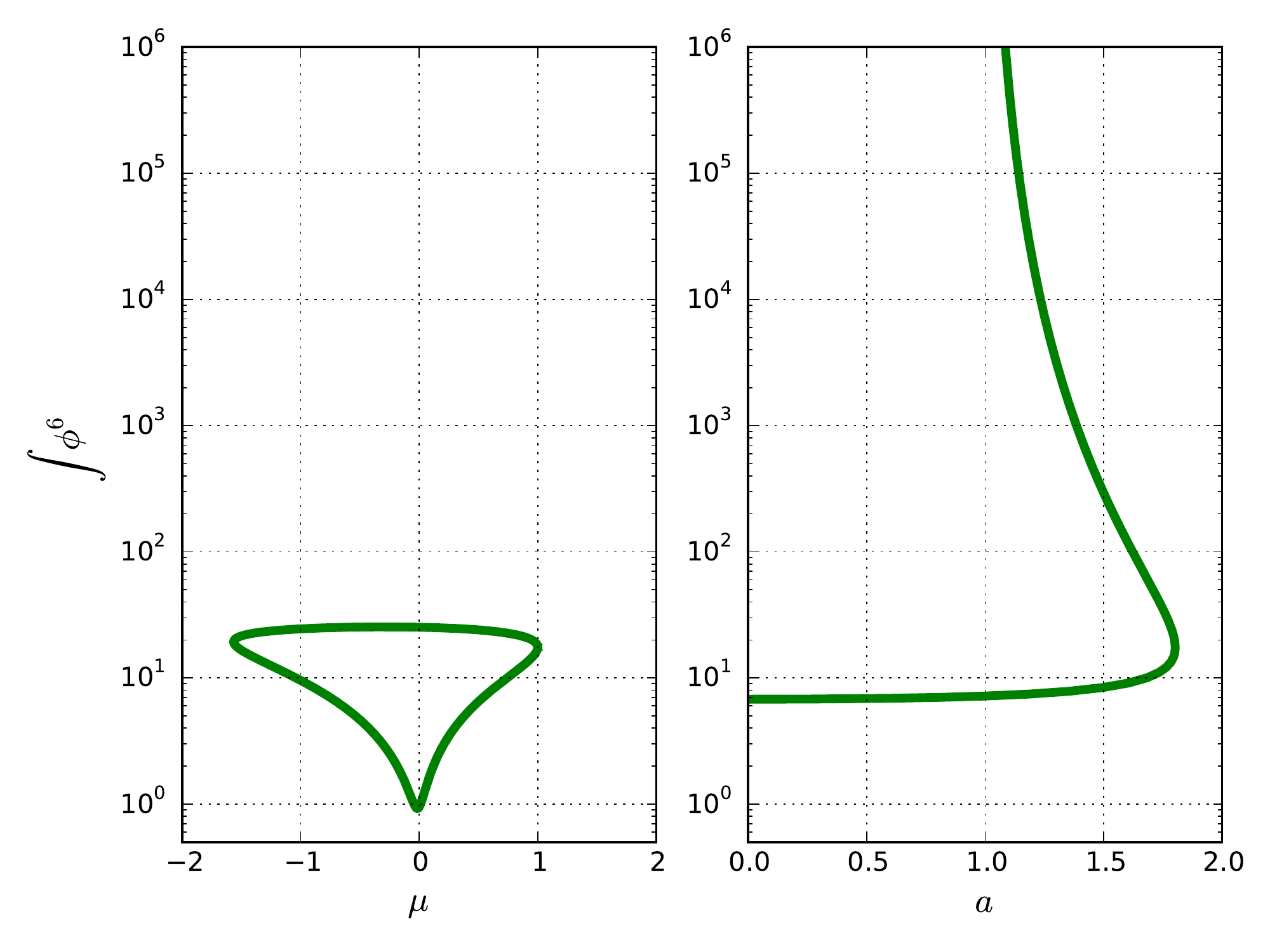}
\caption{Volume of solutions on an $S^1\times M^2$ 
with $R=-0.1$
as the size of the TT tensor and 
as the mean curvature are varied.
The left and right-hand graphs correspond
with the vertical and horizontal gray
dashed lines of Figure \ref{fig:S1-Yneg-mult} respectively.}
\label{fig:S1-Yneg-folds}
\end{figure}

We start by examining solutions on $S^1\times T^2$,
and it turns out that the effects of the two TT tensors $\hat \sigma$ (corresponding to the parameter $\eta$) and $\bar \sigma$ (corresponding to the parameter $\mu$) are quite different.  In particular, if $\eta=0$
and hence $\sigma=\mu\bar\sigma$, it is easy to see that 
$\phi\equiv|\mu|^{1/6}$ and $w=\mu a \sin(s)$ solve system 
\eqref{eqn:S1system}.  These exact solutions are among those
discussed in \cite{Maxwell:2014b}, 
and are the only solutions we were able to find using AUTO.  Hence we obtain 
results consistent with existence and uniqueness for this family,
save for the exceptional case $\mu=0$ where 
the solution degenerates to zero volume.

On the other hand, fixing $\mu=0$, Figure \ref{fig:S1-Yzero-mult} indicates 
the multiplicity of solutions
found on $S^1\times T^2$ with $\tau=1+a\cos(s)$ and 
$\sigma = \eta \hat \sigma$ as the parameters $a$ and $\eta$ are varied.
This computation is an analogue of
the examples of \cite{M11} recalled in Theorem \ref{thm:modelp}, 
except that it involves a family of 
smooth, rather than piecewise constant, mean curvatures.

In the region 
where the mean curvature changes sign,
(i.e., for $a>1$)
we find a fold, indicated by a solid blue line,
and no solutions when the TT tensor is sufficiently large.
Figure \ref{fig:S1-Yzero-folds} indicates how the volume of the solution metric changes
changes as we traverse the gray dashed lines of Figure \ref{fig:S1-Yzero-mult}; the plots, in green,
indicate $\int_{S^1} \phi^6 $, which agrees with volume up to 
an inessential constant factor depending on the second factor of the product manifold.
On the vertical gray dashed line, corresponding to Figure \ref{fig:S1-Yzero-folds} (left-hand side), 
the sign-changing mean curvature is fixed 
and the size of the TT tensor is varied. When the TT tensor is sufficiently large a fold appears and there are no solutions, and as it is
decreased to zero there are two solutions, one heading to zero volume and the other blowing up.  
The horizontal gray dashed line of Figure \ref{fig:S1-Yzero-mult}
corresponds with Figure \ref{fig:S1-Yzero-folds} (right-hand side)
and we again observe the fold and a branch where the volume blows up.
It is difficult to ascertain from this graph the precise value of $a$
where the blowup occurs, and computationally we found it difficult
to approach the singularity.  We find, however, that near the singularity
the value of $\int \phi^6$ along this line is in reasonable
agreement with a growth rate $\sim (a-1)^{-5.14}$.  In later related
examples we find more conclusive evidence of blowup at $a=1$,
so we infer this is the case here as well.  That is, there is
a transition at $a=1$, the threshold of
sign-changing mean curvatures.  The red dashed lines of Figure \ref{fig:S1-Yzero-mult} indicate locations where we have inferred blowup occurs.
The results we observe here are completely parallel with
the prior analytical results of \cite{M11} found for a more
restrictive mean curvature.

Recalling that the only analytical results for sign-changing mean curvatures are available in the Yamabe-null case, we now consider the effect of changing the Yamabe class in these computations.  We first consider
the Yamabe-positive case by adjusting
the previous computation by setting $R=0.001$ in system \eqref{eqn:S1system}
(e.g. by working on an appropriate $S^1\times S^2$). As in the Yamabe-null setting, when working with families of seed data with 
$\sigma=\mu\bar\sigma$ we found tame behavior (one solution
was found for each parameter).  On the other hand, for seed data with 
$\sigma=\mu\hat\sigma$ and $\tau=1+a\cos(s)$ the situation is more
complicated.
Figure \ref{fig:S1-Ypos-mult} indicates the multiplicity of solutions found
for this family and can be compared directly with its
Yamabe-null
counterpart, Figure \ref{fig:S1-Yzero-mult}.  The region of zero
solutions has vanished and we find solutions exist always.  However,
in a region near the original Yamabe-null fold, we find two folds
and a narrow region of
multiple solutions in between.  Figure \ref{fig:S1-Ypos-folds} shows
the effect of traversing along the dashed gray lines of
Figure \ref{fig:S1-Ypos-mult} and indicates how the family of solutions
in this case corresponds with the Yamabe null families in Figure \ref{fig:S1-Yzero-folds}.  Note that the various blowup phenomena found
in the Yamabe-null case have vanished.
The observed fold can be thought of as a purturbation of the situation at $R=0$, and separate computations show that as $R$ is pulled away from zero and approaches, e.g., $R=1$ the volume curves in Figure~\ref{fig:S1-Ypos-folds} stablize further, and the doubling back behavior vanishes.

Turning to the Yamabe-negative case we set $R=-0.1$ in equations \eqref{eqn:S1system} and use $\eta=0$ since we do not have an equivalent
for $\hat\sigma$ for Yamabe-negative seed data.
Thus we use $\mu$ to scale the size of the TT tensor, and 
unlike the Yamabe-positive and -null cases when using
$\sigma = \mu \overline \sigma$, we find interesting results;
Figure \ref{fig:S1-Yneg-mult} shows
the number of solutions found when using mean curvatures of the form $\tau = 1 + a \cos(s)$.  Note that, unlike the parameter $\eta$, 
system \eqref{eqn:S1system} does not have even symmetry with 
respect to $\mu$ and hence our computations involved values of $\mu$ with both signs.  As the mean curvature is made increasingly far-from CMC we find a fold, and subsequently no solutions.  Just as in the Yamabe-null case, the second branch of solutions blows up at $a=1$, the value of $a$ that transitions from constant-sign to sign-changing mean curvatures; 
(Figure \ref{fig:S1-Yneg-folds}, right-hand side).  On the other hand,
far enough into the far-from-CMC regime we were unable to find solutions 
of system \ref{eqn:S1system}. This is perhaps surprising since in the near-CMC setting one can always find solutions when $\sigma\equiv 0$
(unless $\tau\equiv 0$ as well), and indeed solutions at $\sigma\equiv 0$
are a hallmark of Yamabe-negative CMC seed data.  Instead, we find that 
at $\sigma\equiv 0$, as $a$ is increased to make the solution far-from-CMC,
there is a fold around $a=2$ and no solutions were encountered beyond this point.
The absence of solutions appears to be loosely associated with the behavior when $\mu=0$ (i.e. $\sigma\equiv 0$), although one notes that the tip of the `nose' on the blue fold line of Figure \ref{fig:S1-Yneg-mult} does not lie on the line $\mu=0$.
Therefore, there are values of $a$ where no solutions exist at $\mu\equiv 0$, but for which solutions exist for certain values of $\mu\neq 0$.

\subsection{Constant-sign mean curvatures}

We now examine excursions into the far-from CMC regime
using mean curvatures of the form $\tau = \xi^a$, where
$\xi$ is a positive function. Starting again with 
$S^1$-dependent data of the form of the previous section, 
but now with mean curvature
$\tau(s) = (\frac 2 3 - \frac 1 3 \cos(s))^a$, we were
only able to find a single solution of the constraint equations
for all choices of $a$, $\mu$ and $\eta$, except (as is expected) 
when $\sigma\equiv 0$ in the Yamabe non-negative case. We can
understand the tame behavior we observed by appealing to 
the limit equation of Theorem \ref{thm:limiteq}.
For $S^1$-symmetric solutions of the 
$S^1$-dependent data we consider, the limit equation becomes
\begin{equation}\label{eq:S1DepLimitEqn}
  \left(\frac{W'}{N}\right)' = \alpha \left|\frac{W'}{N}\right|\frac{\tau'}{\tau}
\end{equation} 
and it is straightforward to show that this admits no solutions on $S^1$.

\begin{theorem}\label{thm:noLimitSolns}
Let $\tau>0$ be in $C^1$. Then there are no nontrivial $C^2$ solutions $W$ 
of the $S^1$-dependent limit equation \eqref{eq:S1DepLimitEqn}.
\end{theorem}
\begin{proof}
Suppose $W$ is a nontrivial solution, and consider a maximal interval on which 
$W'$ does not vanish.  On this interval we have
\begin{equation}
\log(W')' = k \log(\tau)'
\end{equation}
where $k=\alpha$ if $W'>0$ on the interval, and $k=-\alpha$ if $W'<0$.
Hence
\begin{equation}\label{eq:we-know-what-W-is}
W'  = c\tau^k
\end{equation}
for some constant $c\neq 0$.  At the endpoints of the interval $W'$ tends to zero.
But the right-hand side of equation \eqref{eq:we-know-what-W-is} is uniformly bounded
away from zero.
\end{proof}

Strictly speaking, Theorem \ref{thm:limiteq}
does not apply in our setting
because of the presence of conformal Killing fields. 
We expect, however, that one can use techniques 
found in \cite{M11} to adapt the main theorem of \cite{DGH10} 
to this specific 
family of data to conclude that the nonexistence of solutions to
\eqref{eq:S1DepLimitEqn} implies existence of ($S^1$-symmetric)
solutions.

\begin{figure}
  \includegraphics[height=\plotheight]{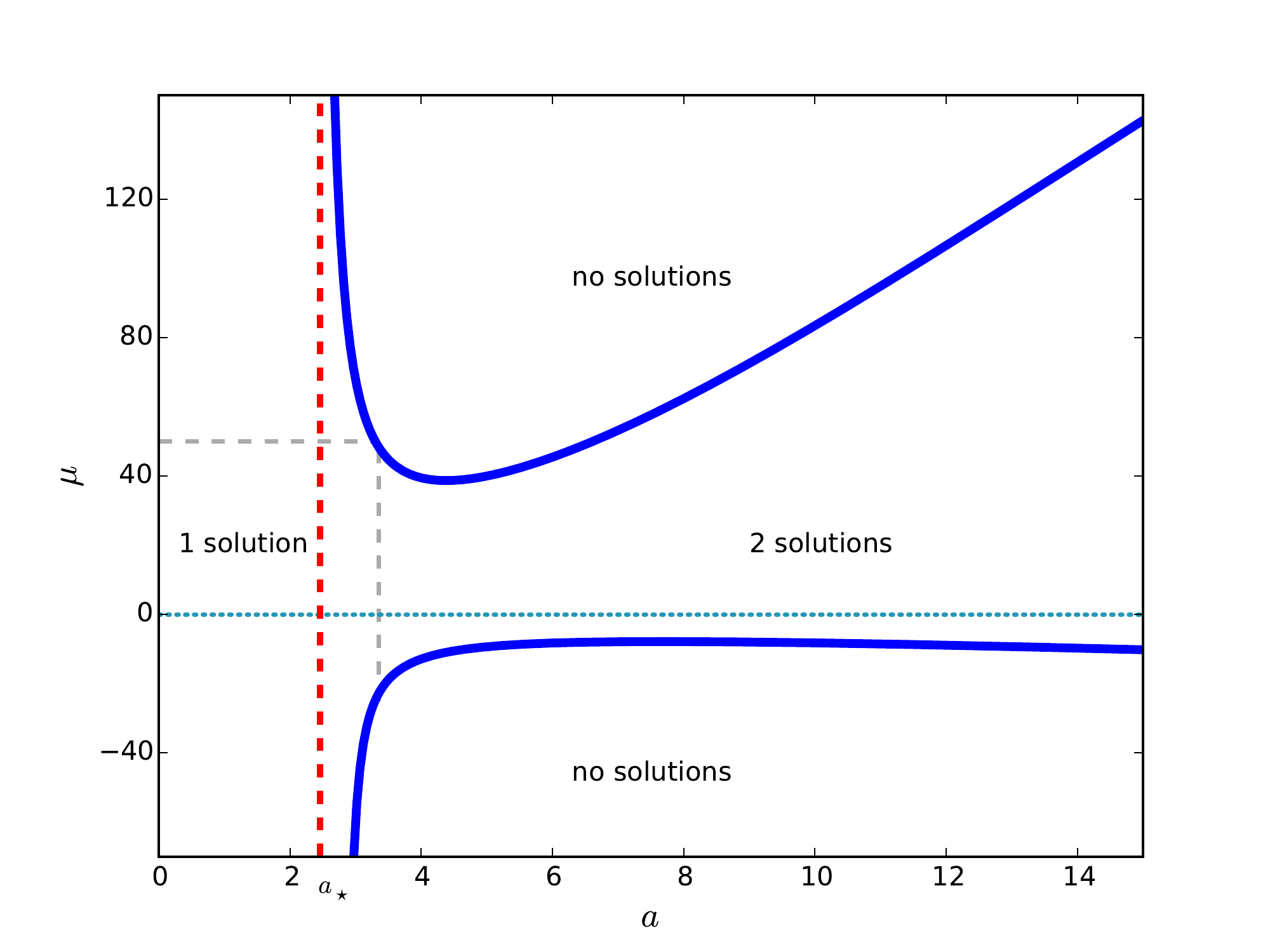}
\caption{  \label{fig:LatDepSolvability}
Multiplicity of solutions found on an $S^2\times S^1$ 
with positive scalar curvature $R=1$. Seed data: $\tau=(\frac{2}{3}+\frac{1}{3}\cos(\phi))^a$ and $\sigma = \mu \overline \sigma$.
The solid blue line is a computed fold, whereas the red dashed line
indicate locations where blowup is inferred. At the dotted blue
line at $\mu=0$ there is a zero volume solution which should be discounted. 
The gray dashed lines 
are discussed in Figure \ref{fig:LatDep-folds}.}

\includegraphics[height=\plotheight]{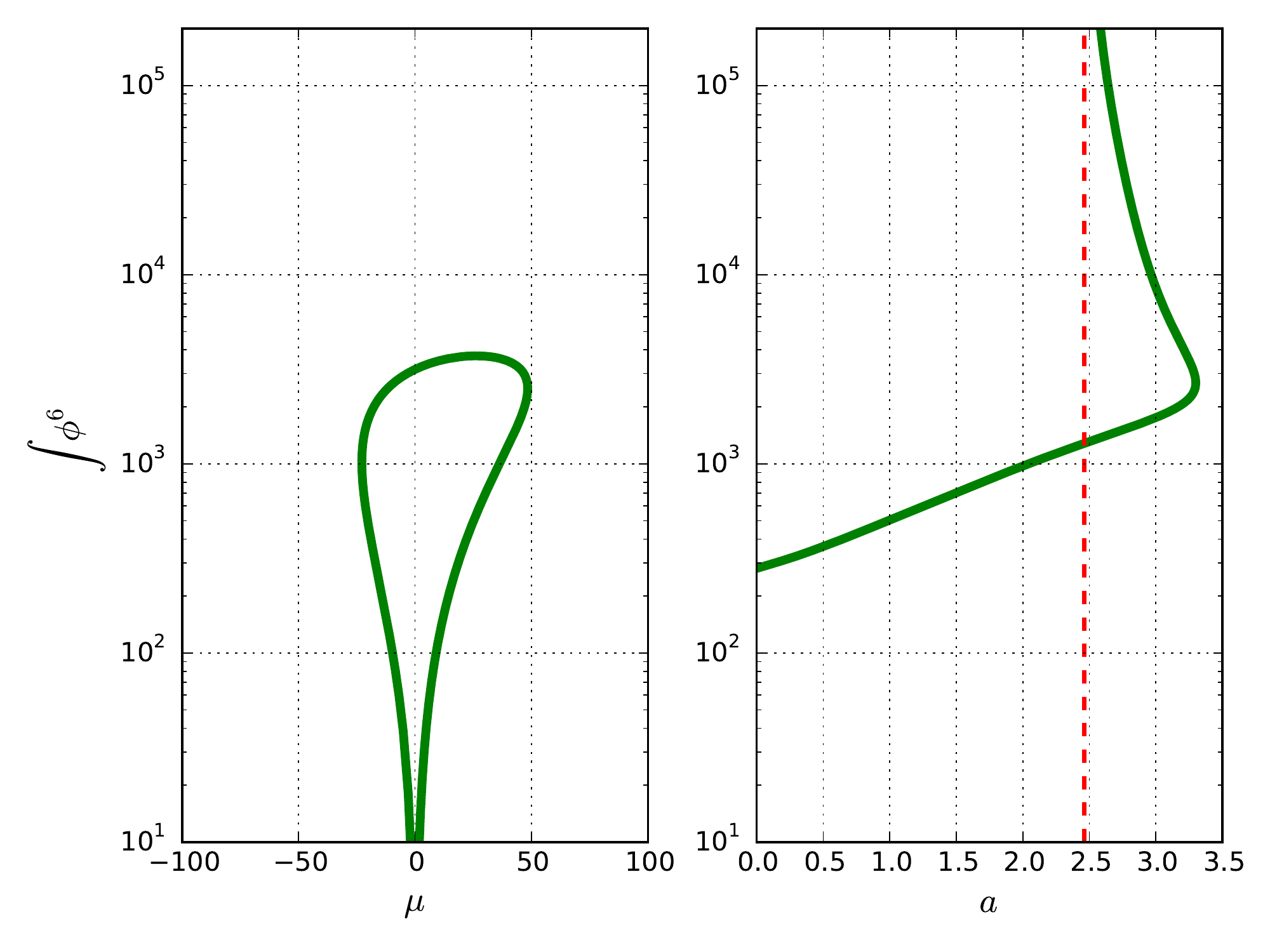}

\caption{\label{fig:LatDep-folds}
Volume of solutions on an $S^2\times S^1$ 
with $R=1$
as the size $\mu$ of the TT tensor and 
as the mean curvature $\xi^a$ is varied.
The left- and right-hand graphs correspond
with the vertical and horizontal gray
dashed lines of Figure \ref{fig:LatDepSolvability} respectively.}
\end{figure}

\begin{figure}
\includegraphics[height=\plotheight]{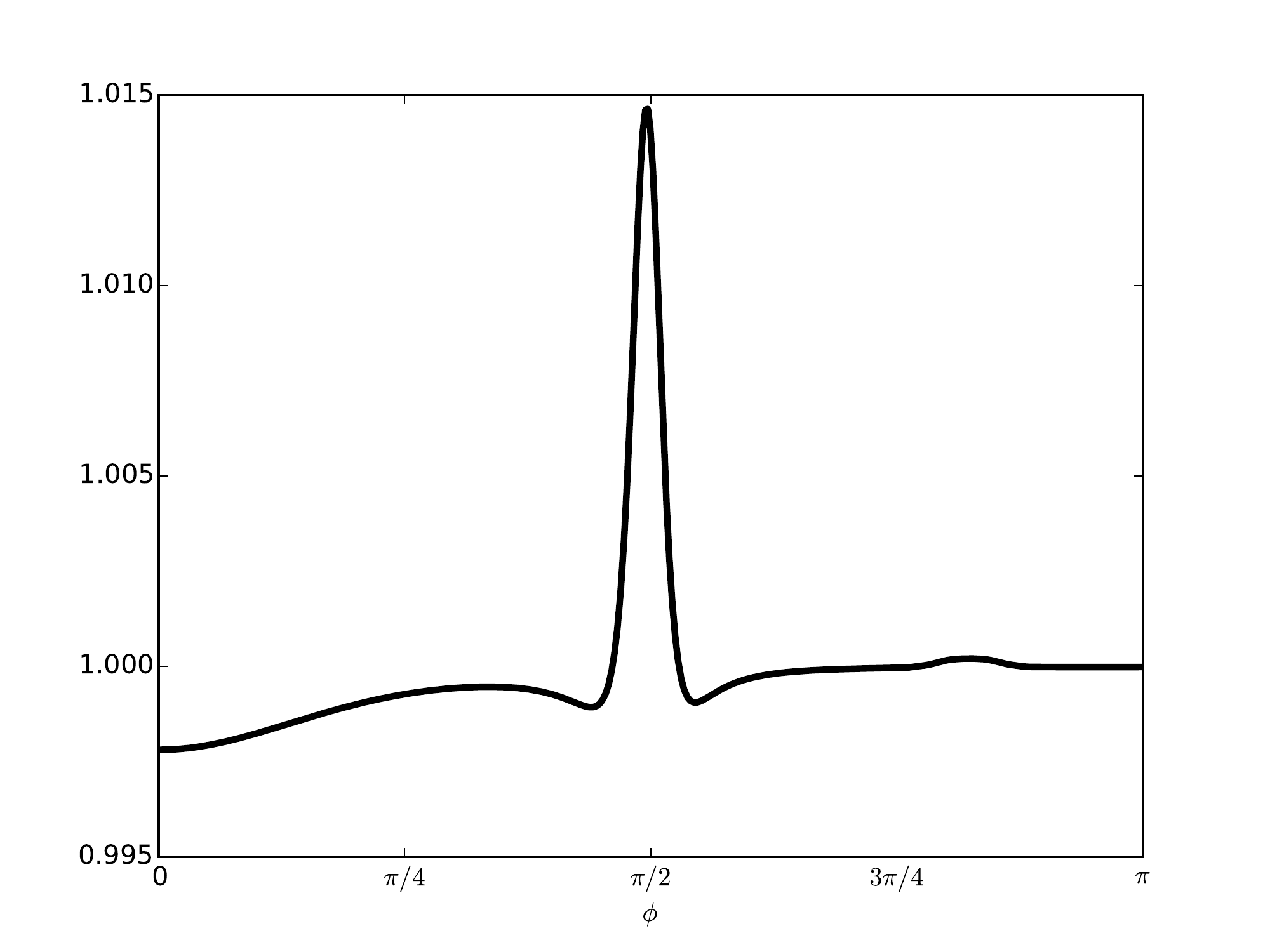}

\caption{\label{fig:LatDep-limit-eq}
Ratio of $\sqrt{\frac{2}{3}}\varphi^6 \tau$ to
$\left|\frac{1}{2N}\ck W\right|$ as a function
of latitude $\phi$ for the large
solution of the conformally-parameterized
constraint equations at $a=2.46$, $\mu=1$ in Figure 
\ref{fig:LatDepSolvability}.  The ratio is nearly
1, indicating that the vector field $W$ of the solution
is nearly a solution of the limit equation \eqref{eq:limitEqn}.}
\end{figure}

The simple behavior seen for $S^1$-dependent data does not
hold generally, however.  Throughout the remainder of this 
section we use conformal seed data of the following type:
\begin{itemize}
    \item The manifold is $S^2\times M_2$ where
    $M_2$ is one of $S^1$,or $H^2$.
    We use $\phi\in(0,\pi)$ for a latitude parameter on $S^2$.
    \item The mean curvature $\tau=\xi^a$ depends only on $\phi$,
    where $0\le \xi \le 1$ and $\xi(\phi)=1$ somewhere.
    In practice we used
    \[
    \xi(\phi) = \frac 2 3 -\frac 1 3 \cos(k\phi)
    \]
    with $k=1$ or $2$. 
    \item The TT tensor is $\sigma = \mu \overline \sigma$,
    where $\overline \sigma$ was introduced in equation 
    \eqref{eq:barsigma}.
    \item The lapse density is $N=1/2$.
\end{itemize}

As a first example, consider
data on $S^2\times S^1$.
The reduced conformal constraint equations for 
latitude-dependent data are more
complicated than those for $S^1$-dependent data; the main
differential operators are
\begin{align*}
  \Lap f &= \frac1r \left(f''+ \cot\phi f'\right), \,\,\,\,\, \textrm{ and}\\
  \div \mathcal{L}W &= (3 w'' + 3 \cot\phi w' +\frac1{2r}(1-3 \cot^2\phi)w)d\phi
\end{align*} where $f=f(\phi)$ and $W=w(\phi)d\phi$.  We seek solutions
of the constrant equations 
\eqref{Eq:HamConstr}-\eqref{Eq:MomConstr} the form $\varphi=\varphi(\phi)$
and $W=w(\phi)d\phi$ supplemented with boundary conditions
$\varphi'=0$ and $w'=0$ at $\phi=0$ and $\phi=\pi$ needed to ensure regularity.  An additional boundary condition $w=0$ at $\phi=0,\pi$
is effectively enforced by the momentum constraint.

Figure \ref{fig:LatDepSolvability} shows the number of solutions found 
on an $S^2\times S^1$ with scalar curvature $R=1$ with 
$\sigma=\mu\overline\sigma$ and 
a relatively simple non-CMC mean curvature
\[
\tau = \left[\frac{2}{3}+\frac12 \cos(\phi)\right]^a
\]
Theorem \ref{thm:ngn2} would apply to this data, except 
for the usual caveat about conformal Killing fields and,
more crucially, the fact that 
the mean curvature violates the non-generic inequality
\eqref{eq:almostckf}. We nevertheless find behaviour consistent with its
conclusions:  multiple solutions when both $a$ is sufficiently large 
and $\mu\neq 0$ is sufficiently small.  Moreover, the transition to the far-from-CMC regime is abrupt, starting at 
$a\approx 2.46$, which we will call $a_*$.  Figure \ref{fig:LatDep-folds}, 
(right-hand side) shows that $a_*$ is associated with a blowup of a branch of solutions, and one expects
there is a solution of the limit equation \eqref{eq:limitEqn}
with $\alpha=1$ at $a=a_*$.  Indeed, the limit equation arises as
the Hamiltonian constraint degenerates to the algebraic equation
\begin{equation}\label{ref:limeq-base}
\sqrt{\frac{2}{3}}\phi^6\tau = \left| \frac{1}{2N}\ck W \right|,
\end{equation}
which can then be substituted back into the momentum constraint.
Figure \ref{fig:LatDep-limit-eq} shows the ratio of the two sides of 
equation \eqref{ref:limeq-base} for the larger of the two
solutions at the point
$a=2.46$m $\mu = 1$  in Figure \ref{fig:LatDepSolvability}. The ratio
is nearly 1, so the solution vector field $W$ at that point is nearly
a solution of the limit equation.

 The proofs of 
Theorems  \ref{thm:Nguyen} and \ref{thm:ngn2} demonstrate the existence
of multiple solutions for $a$ large and $\mu$ small by showing
that at $\mu=0$ there is both a zero volume solution and a true 
solution, and
that perturbing off of these yields two solutions; this mechanism is illustrated in Figure \ref{fig:LatDep-folds} (left-hand side).  Recall that
solutions with $\mu=0$ are impossible in the near-CMC setting for Yamabe-positive seed data such as this, and
at the dotted line at $\mu=0$ of Figure \ref{fig:LatDepSolvability} 
there is one less solution than at neighboring points 
(so there is no solution
to the left of the singularity at $a=a_*$ and just one solution to the right).

\begin{figure}
  \includegraphics[height=\plotheight]{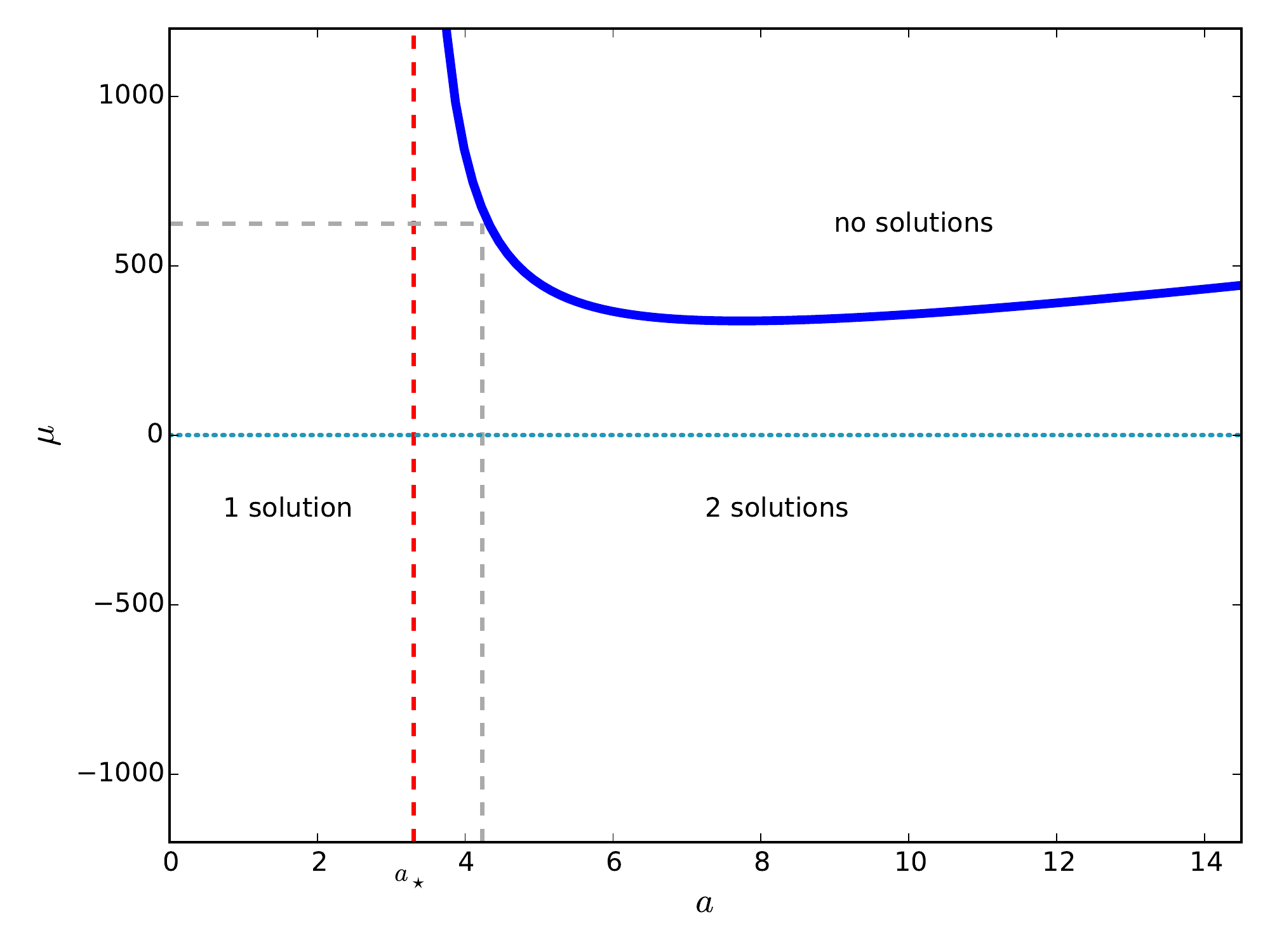}
\caption{  \label{fig:LatDep-NTC-Solvability}
Multiplicity of solutions found on an $S^2\times S^1$ 
with positive scalar curvature $R=1$. Seed data: 
$\tau=(\frac{2}{3}+\frac{1}{3}\cos(2\phi))^a$ 
and $\sigma = \mu \overline \sigma$; note the $2$ in the argument of $\cos$.
The solid blue line is a computed fold, whereas the red dashed lines
indicate locations where blowup is inferred. 
At the dotted blue line there is a zero-volume solution,
which should be discounted.
The gray dashed lines 
are discussed in Figure \ref{fig:LatDep-NTC-folds}.}

\includegraphics[height=\plotheight]{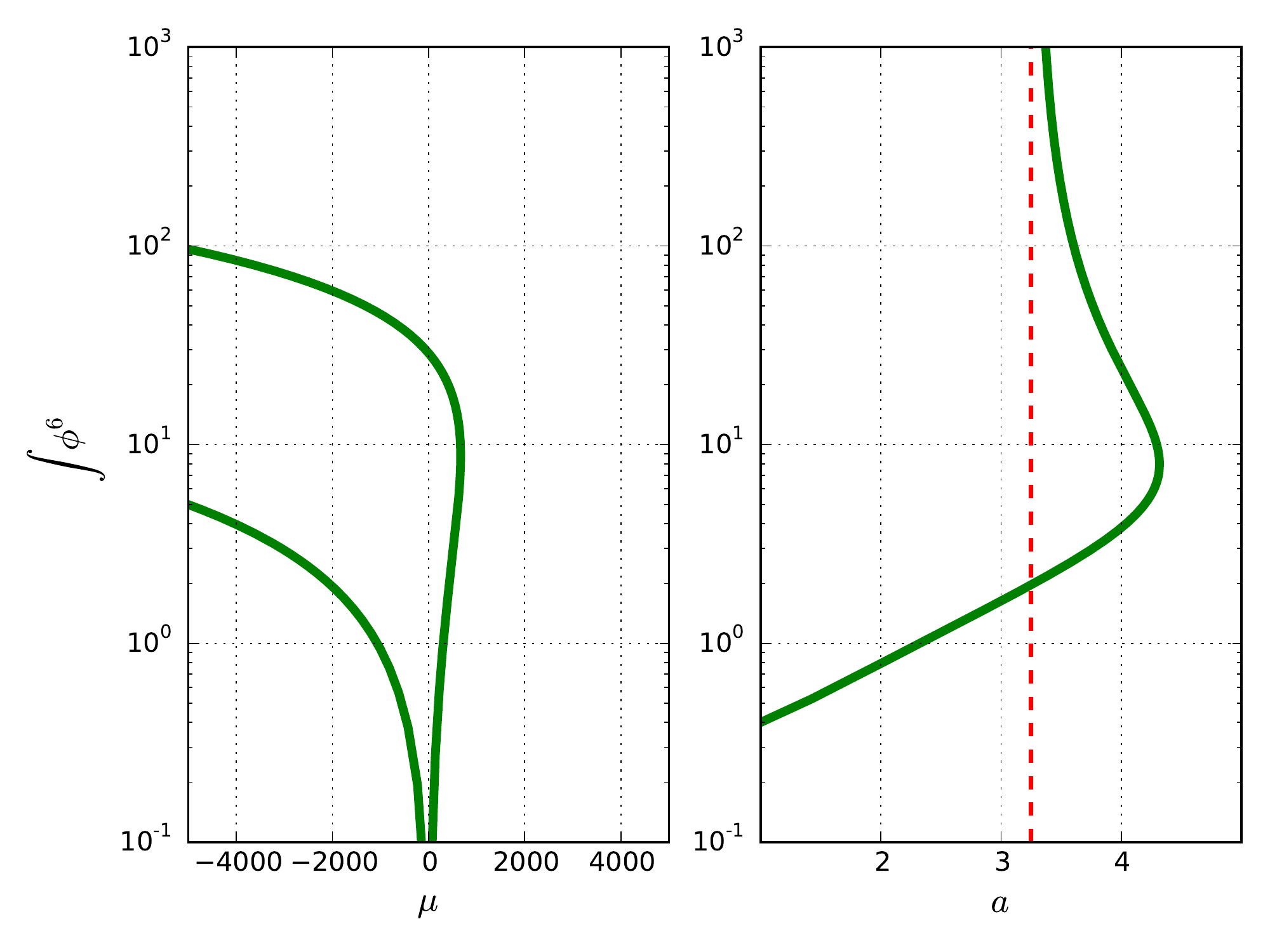}

\caption{\label{fig:LatDep-NTC-folds}
Volume of solutions on an $S^2\times S^1$ 
with $R=1$
as the size $\mu$ of the TT tensor (left-hand side) and 
as the mean curvature $\xi^a$ (right-hand side) are varied.
The left- and right-hand graphs correspond
with the vertical and horizontal gray
dashed lines of Figure \ref{fig:LatDep-NTC-Solvability} respectively.}
\end{figure}

As a test of the robustness of these results, 
consider the same conformal seed data as the previous
example, except that the mean curvature is now
\[
\tau = \left[\frac{2}{3}+\frac12 \cos(2\phi)\right]^a.
\]
Figure \ref{fig:LatDep-NTC-Solvability}
illustrates the multiplicity of solutions we found as we varied
the parameters $\mu$ and $a$.  Here again we find a sharp
transition to the far-from-CMC setting, now at $a=a_*\approx 3.3$. 
There is blowup associated with this transition (Figure \ref{fig:LatDep-NTC-folds}, right-hand side), and again we presume
there is a solution to the limit equation with $\alpha=1$ and $a=a_*$.
For $a>a_*$ we find a solution
at $\mu=0$ in addition to the zero volume solution (Figure \ref{fig:LatDep-NTC-folds}, left-hand side), 
and hence there two solutions for $\mu$ sufficiently small. However,
we find an apparent difference between scaling $\mu$ large and positive
versus large and negative for this seed data.  
For $\mu$ positive and large enough 
(depending on $a$) there are no solutions, just as 
was the case in Figure \ref{fig:LatDepSolvability}.
But for $\mu$ large and negative we were unable to find a fold 
and a consequential transition to zero solutions.  Instead, the
two solutions remain well-separated in volume as $\mu$ is made 
large (Figure \ref{fig:LatDep-NTC-folds}, left-hand side).  We
cannot rule out the possibility that the two branches
eventually merge, but this was not the case out to
$\mu=-2500000$.  Although this data violates inequality
\ref{eq:almostckf} of Theorem \ref{thm:ngn2}, our observations are nevertheless 
consistent with its conclusions.
In particular, unlike Theorem \ref{thm:Nguyen}, Theorem
\ref{thm:ngn2} does not predict non-existence for large TT tensors,
and it is conceivable that \ref{thm:ngn2} holds 
more generally for mean curvatures not satisfying inequality 
\ref{eq:almostckf}.

\begin{figure}
  \includegraphics[height=\plotheight]{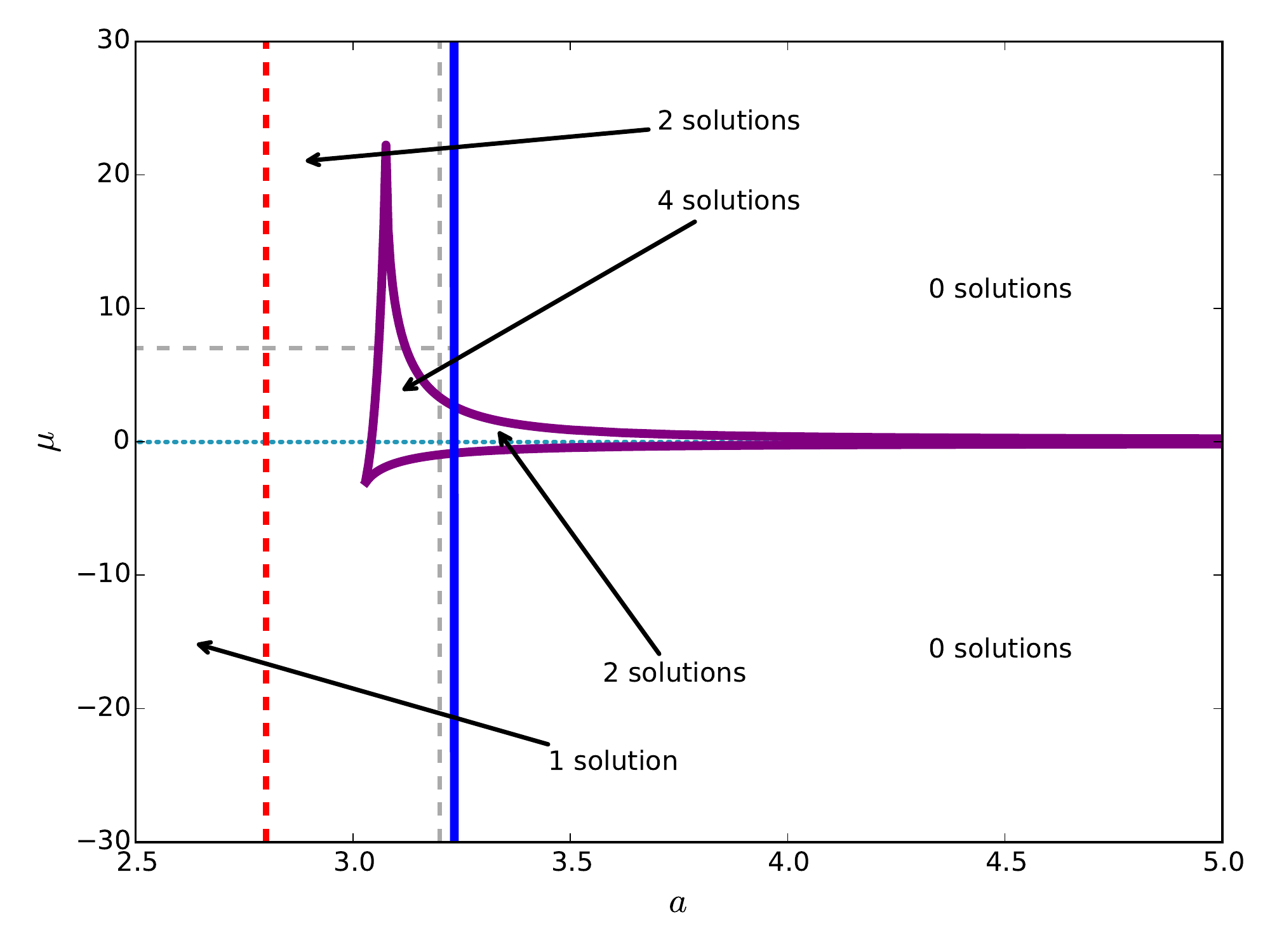}
\caption{  \label{fig:S2H2-Solvability-Pos}
Multiplicity of solutions found on an $S^2\times H^2$ 
with positive scalar curvature $R=0.1$.
Seed data: $\tau=(\frac{2}{3}+\frac{1}{3}\cos(\phi))^a$ 
and $\sigma = \mu \overline \sigma$.
The solid blue and purple lines are computed folds, 
whereas the red dashed line
indicate locations where blowup is inferred. 
At the blue dotted line at $\mu=0$ one solution has
zero volume and should be ignored.  
Solutions along the gray dashed lines
are discussed in Figures \ref{fig:S2H2-ManyFolds}
and \ref{fig:S2H2-MuLoops}.}

\includegraphics[height=\plotheight]{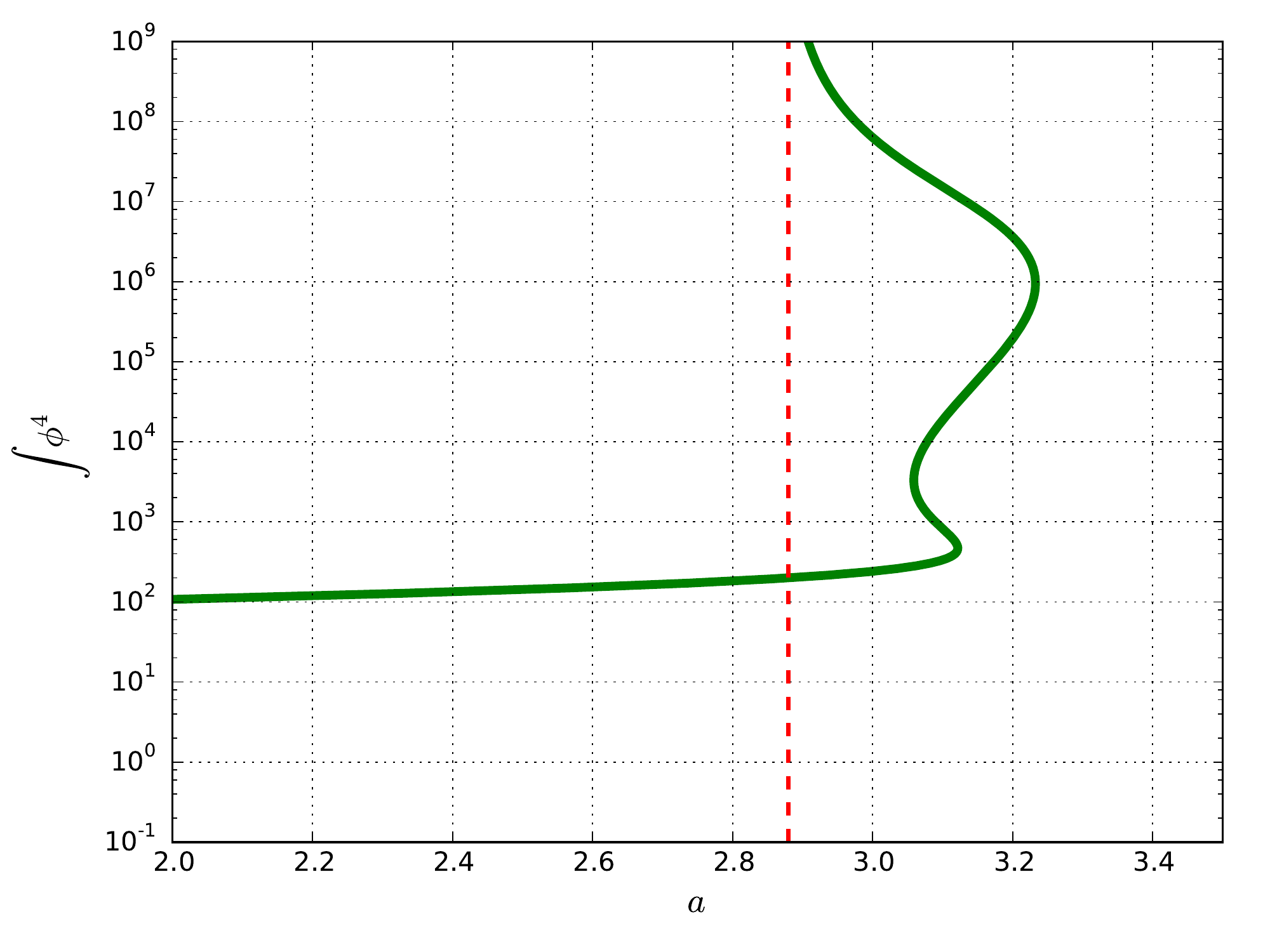}
\caption{\label{fig:S2H2-ManyFolds}
Volume of solutions on an $S^2\times H^2$ 
with $R=0.1$ 
as the mean curvature $\xi^a$ (right-hand side) is varied.
Solutions correspond with horizontal gray
dashed line at $\mu=7$ in 
Figure \ref{fig:S2H2-Solvability-Pos}. At $a=3.1$, four solutions are found for a single
conformal seed dataset.}
\end{figure}

\begin{figure}
  \includegraphics[height=\plotheight]{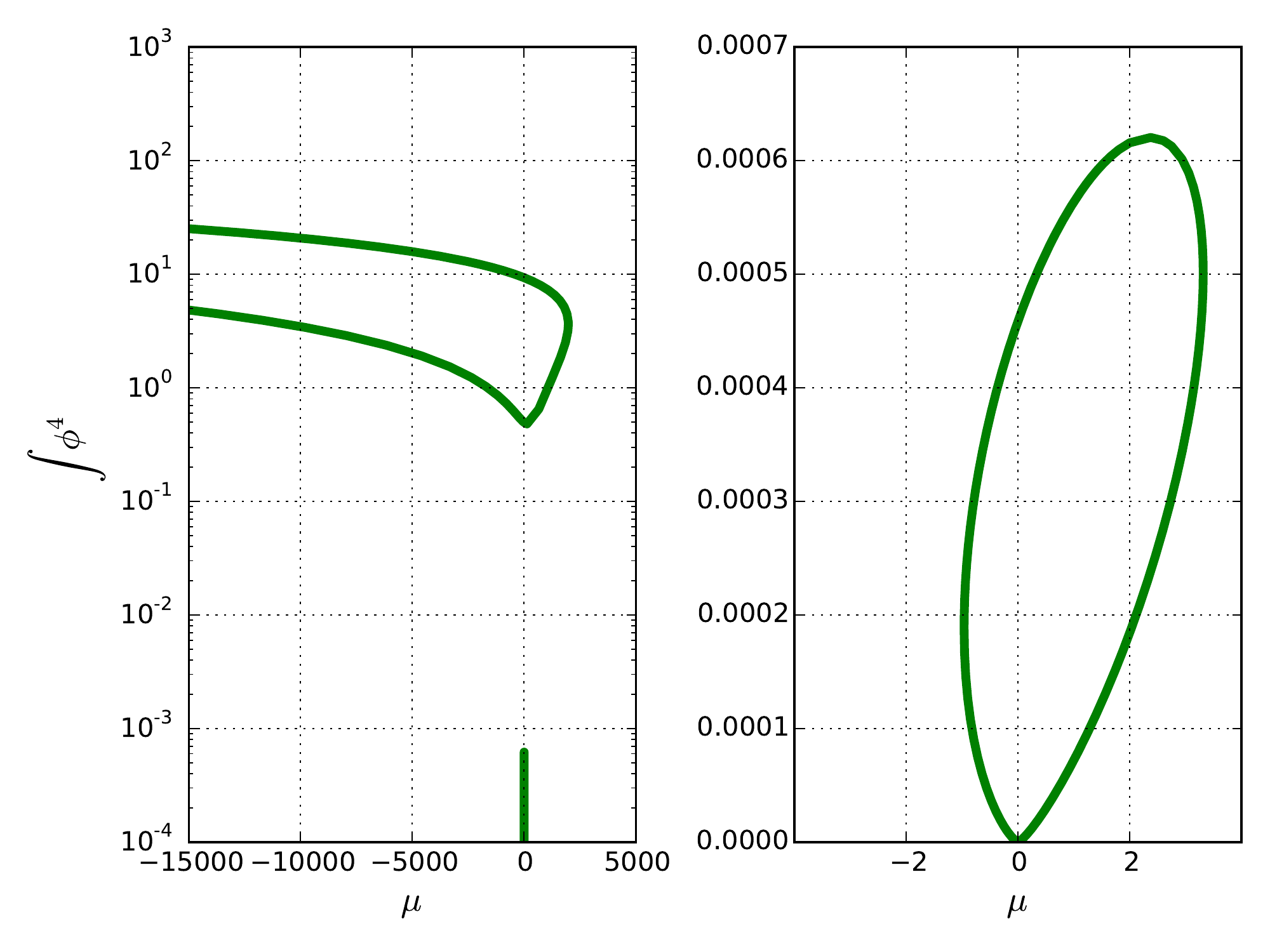}
  \caption{  \label{fig:S2H2-MuLoops}
Left-hand side: volume of solutions on an $S^2\times H^2$ 
with $R=0.1$ as the size $\mu$ of the TT tensor is varied.
Solutions correspond with the vertical gray
dashed line at $a=3.2$ in 
Figure \ref{fig:S2H2-Solvability-Pos}.
Right-hand side shows a detail of the small loop near $\mu=0$.
The large upper loop eventually closes at $\mu\approx-800000$.}
\end{figure}

The previous two examples involve Yamabe-positive data.
To explore the other Yamabe classes we consider latitude-dependent
data on $S^2\times H^2$; by varying the size of the round $S^2$ factor
we can obtain any desired constant scalar curvature.
Before looking at the other Yamabe classes, however, we 
remark that even for Yamabe-positive
data of this type we find differences from what
was observed for $S^2\times S^1$.  Figure
\ref{fig:S2H2-Solvability-Pos}
illustrates the multiplicity of solutions
found on an $S^2\times H^2$ with $R=0.1$, where the
seed data has $\sigma = \mu \overline \sigma$
and
\[
\tau = \left[\frac{2}{3}-\frac13 \cos(\phi)\right]^a.
\]
This seed data is comparable to that used for 
Figure \ref{fig:LatDepSolvability}, but there
are some fine differences in the results.  Again we find
a sharp transition to the far-from-CMC setting, now at
$a=a_*\approx 2.8$.
However, for $a>a_*$ the multiplicity of solutions
is a bit complicated.  There is a primary fold roughly at
$a=3.23$ in the far-from-CMC regime.  Beyond this point, the behavior
is similar to that of Figure \ref{fig:LatDepSolvability}, with
no solutions when $\mu$ is sufficiently large and two solutions 
when $\mu\neq 0$ is sufficiently small.  Between $a=a_*$ and $a=3.23$,
however, the situation has changed from that of Figure \ref{fig:LatDepSolvability}. We found variously between 1 and 4 solutions, 
and no convincing evidence
that there are no solutions when $\mu$ is sufficiently
large at, e.g., $a=2.87$.

Figure \ref{fig:S2H2-ManyFolds},  illustrates how the the various 
folds appear as the dashed line at $\mu=7$ in Figure \ref{fig:S2H2-Solvability-Pos} is traversed. First the right-most
purple fold is encountered, then the left-most and finally the blue fold 
around $a=3.23$ is hit before heading off to the singularity at $a=a_*$.  
Conversely, Figure \ref{fig:S2H2-MuLoops} illustrates the various
solutions along the line $a=3.2$, and we see two disconnected loops 
(the upper branches in Figure \ref{fig:S2H2-MuLoops}, left-hand side,
 rejoin at
$\mu\approx -800000$).  As was discussed previously for Figure 
\ref{fig:LatDep-folds}, Figure \ref{fig:S2H2-MuLoops}
illustrates how one can visualize the 
multiple solutions near $\mu=0$ as perturbations of the usual zero solution at $\mu=0$ and the additional three non-trivial solutions we found there.

\begin{figure}
\ifjournal
\hspace*{-2cm}\includegraphics[height=45mm]{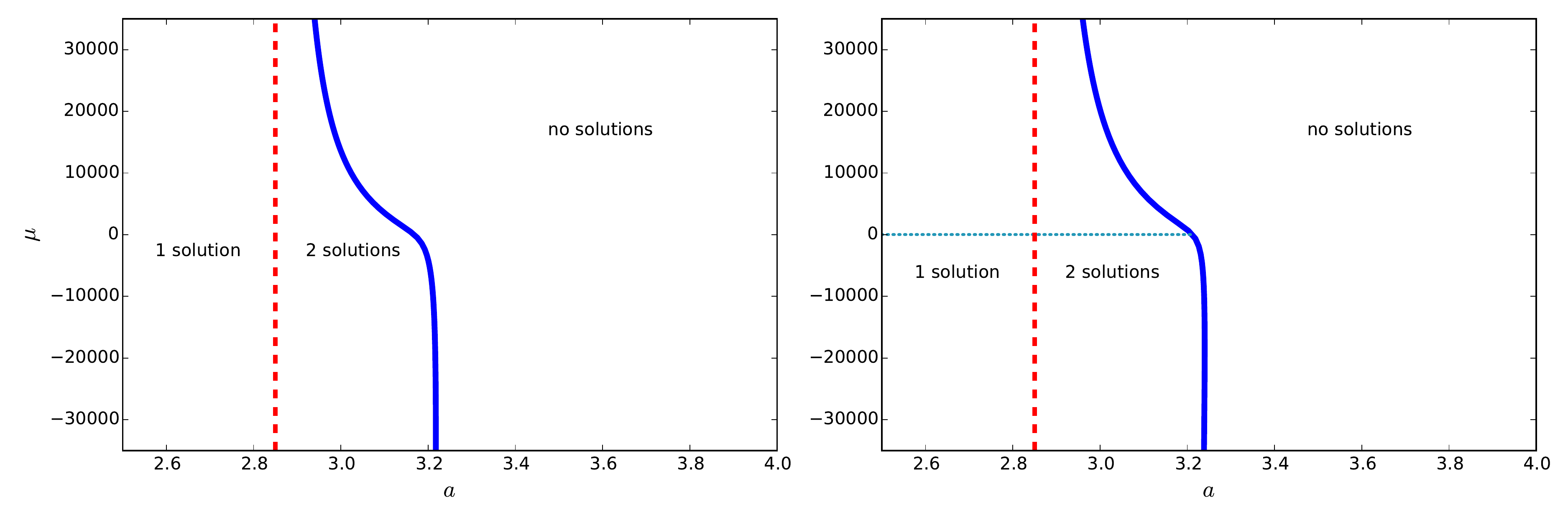}
\else
\includegraphics[height=55mm]{graphs/S2H2-solvability-nonpos.pdf}
\fi
\caption{\label{fig:S2H2SolvabilityNonpos}
Multiplicity of solutions found on an $S^2\times H^2$ 
with scalar curvatures $R=-0.2$ (left-hand side) and
 $R=0$ (right-hand side).
Seed data: $\tau=(\frac{2}{3}-\frac{1}{3}\cos(\phi))^a$ 
and $\sigma = \mu \overline \sigma$.
The solid blue lines are computed folds, 
whereas the red dashed lines
indicate locations where blowup is inferred. 
The folds are comparable with
the blue fold of Figure \ref{fig:S2H2-Solvability-Pos},
which has a similar shape when drawn at this scale.
On the right-hand side ($R=0$), the dotted blue line indicates a zero
solution which should be discounted. Fine details near $\mu=0$
in the region $2.8<a<3.2$ are potentially unresolved; see
Figure \ref{fig:S2H2SigmaZero} (middle).
}
\end{figure}

Turning to non-Yamabe-positive seed data,
Figure \ref{fig:S2H2SolvabilityNonpos} 
illustrates the number of solutions found
with the same conformal data as in Figure \ref{fig:S2H2-Solvability-Pos},
except now with $R=-0.1$ and $R=0$.  For this seed data 
there are no applicable far-from-CMC theorems to guide expectations,
and we find
a situation similar to the Yamabe-negative sign-changing
data of Figure \ref{fig:S1-Yneg-mult}.  There is a sharp transition
to far-from-CMC data at $a_*\approx 2.8$, but now solutions vanish
sufficiently far into the far-from CMC regime. The blue fold
from the Yamabe-positive data of 
Figure \ref{fig:S2H2-Solvability-Pos} persists, but the purple
fold that extended out along the line $\mu=0$ has vanished.
The apparent non-existence of solutions far enough into the
far-from-CMC regime violates the conclusion of 
Theorem \ref{thm:ngn2}, and we therefore suspect that the
Yamabe-positive hypothesis of Theorem \ref{thm:ngn2} is essential.

Although the vertical scales for Figures 
\ref{fig:S2H2-Solvability-Pos} and \ref{fig:S2H2SolvabilityNonpos}
are markedly different, it is not the case that we have missed 
a fine feature near the line $\mu=0$ in Figure 
\ref{fig:S2H2SolvabilityNonpos} corresponding to the purple
folds of Figure \ref{fig:S2H2-Solvability-Pos}. Indeed,
Figure \ref{fig:S2H2SigmaZero} illustrates the volume of solutions
along the line $\mu=0$ (i.e. $\sigma\equiv 0$) for the
three cases $R=-0.2$, $R=0$ and $R=0.1$ illustrated in Figures \ref{fig:S2H2-Solvability-Pos} and \ref{fig:S2H2SolvabilityNonpos}.
For far-from-CMC Yamabe-positive seed data (Figure \ref{fig:S2H2SigmaZero}, right-hand side) we find solutions at $\sigma\equiv 0$; these solutions
are not present for near-CMC data and had not been expected
before \cite{Nguyen:2015}.
By contrast, for the Yamabe-negative seed data the 
near-CMC solutions at $\sigma\equiv 0$ vanish once the
mean curvature is sufficiently far-from CMC (Figure \ref{fig:S2H2SigmaZero}, left-hand side).  At the boundary $R\equiv 0$
(Figure \ref{fig:S2H2SigmaZero}, middle)
we find a narrow band in which solutions exist; none in the near-CMC
case (which is expected) and none for the far-from CMC seed data as well.
The narrow band is reminiscent of the well known 
one-parameter family of solutions found 
found for CMC Yamabe-null seed data when $\tau\equiv 0$ and when $\sigma\equiv 0$.  It is also similar to the somewhat analogous one-parameter
families found for particular Yamabe-null seed data in 
\cite{M11} and \cite{Maxwell:2014b}.  What would have been a vertical 
line of solutions in the analytic cases 
has been deformed to the distorted vertical line
of Figure \ref{fig:S2H2SigmaZero}, middle.

The role of solutions at $\sigma\equiv 0$ appears to be important
for understanding the conformal method in the far-from-CMC setting,
and we believe this is a consequence of $\sigma$ being absent
from the limit equation.  Nevertheless, the interplay between
$\sigma\equiv 0$ and $\sigma\not\equiv 0$ is nuanced. 
For example, the breakdown of the existence of solutions along 
$\sigma\equiv 0$
for Yamabe-negative data (Figure \ref{fig:S2H2SigmaZero}, left-hand side)
at $a\approx 3.2$
is not perfectly correlated with nonexistence of solutions for nearby
values of $a$.  This can be seen by inspecting the blue fold of
Figure \ref{fig:S2H2SolvabilityNonpos}, left-hand side.  It
crosses $\mu=0$ at a value $a\approx 3.2$ but there nevertheless
exist certain solutions for values of $a>3.2$ beyond the crossing point,
but not much larger.

\begin{figure}
\ifjournal
\hspace*{-2cm}\includegraphics[height=45mm]{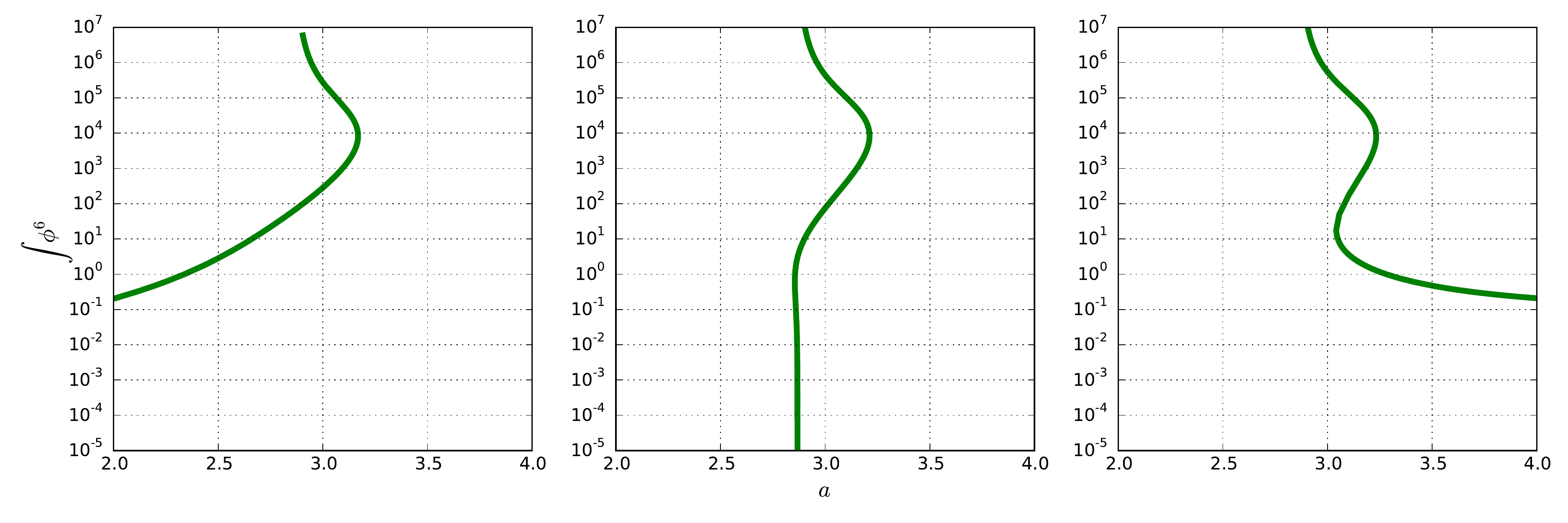}
\else
\includegraphics[height=55mm]{graphs/S2H2-SigmaZero.pdf}
\fi
  \caption{ \label{fig:S2H2SigmaZero}
Volume of solutions along
  $\sigma\equiv 0$ for an $S^2\times H^2$
having $R=-0.2$ (left),
  $R=0$ (middle) and $R=0.1$ (right).  Mean curvature
  $\tau = (\frac23 -\frac 13\cos(\phi))^a$.
Conformal seed data is the same as in Figures \ref{fig:S2H2-Solvability-Pos}
and \ref{fig:S2H2SolvabilityNonpos}.
  }
\end{figure}

Finally, we consider seed data on $S^2\times H^2$ 
with 
$\sigma=\mu \overline \sigma$ and 
\[
\tau = \left[\frac{2}{3}-\frac13 \cos(2\phi)\right]^a,
\]
which can be compared with the seed data used for Figure
\ref{fig:LatDep-NTC-Solvability}.
Figure \ref{fig:S2H2-NTC-solvability-nonneg} illustrates the number of solutions found
for $R=0$ and $R=0.5$, and the outcomes are qualitatively similar 
to those of Figure \ref{fig:LatDep-NTC-Solvability}.
In particular, for $\mu<0$ we do not find evidence of non-existence of solutions. On the other hand, for this same data but $R=-1/2$ we obtain 
multiplicities shown in Figure \ref{fig:S2H2-NTC-solvability-neg}
and find a situation akin to 
what we have seen previously in 
Figure \ref{fig:S1-Yneg-mult} and 
Figure \ref{fig:S2H2SolvabilityNonpos} (left-hand side)
for Yamabe-negative data: no solutions
for $a$ sufficiently large. In an independent computation, not shown,
we found that the location of the blue fold crossing $\mu=0$
(e.g. at $a\approx 108$ in Figure \ref{fig:S2H2-NTC-solvability-neg})
grows as $(-R)^{-13}$ and therefore we believe there is 
indeed a transition at $R=0$ between the two qualitative behaviors
seen here for $R=1/2$ and $R=-1/2$.

\begin{figure}
\ifjournal
\hspace*{-2cm}\includegraphics[height=45mm]{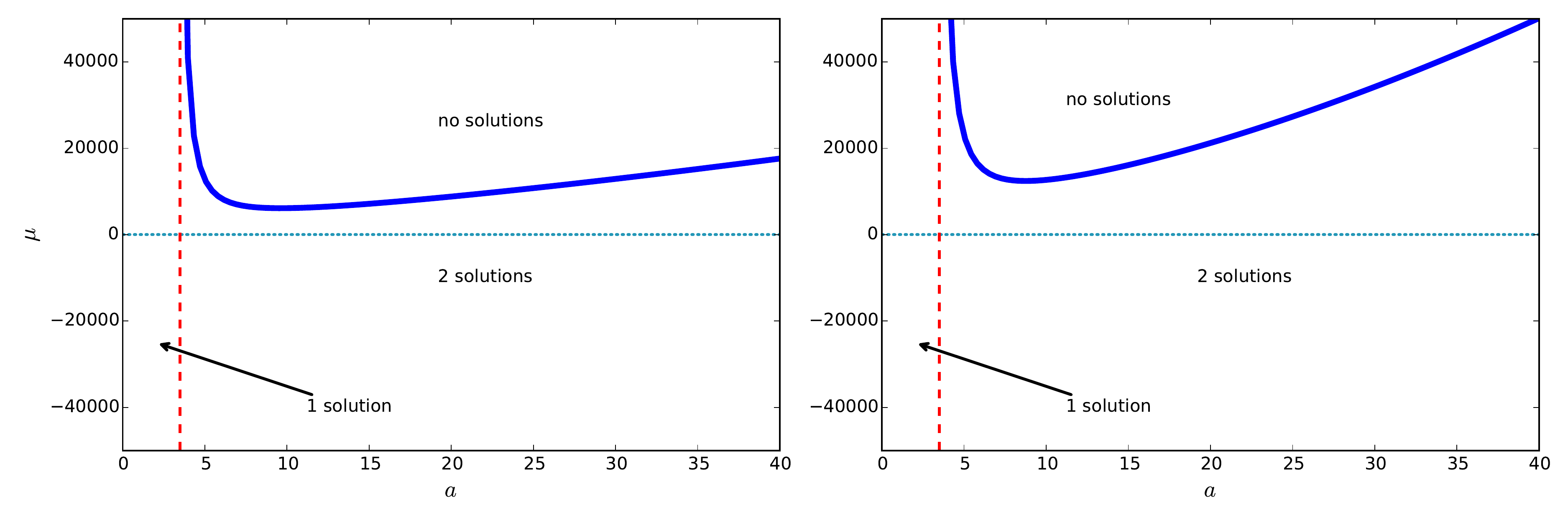}
\else
  \includegraphics[height=55mm]{graphs/S2H2-NTC-solvability-nonneg.pdf}
\fi
\caption{  \label{fig:S2H2-NTC-solvability-nonneg} 
Multiplicity of solutions found on an $S^2\times H^2$ 
with scalar curvatures $R=0$ (left-hand side)
and $R=0.5$ (right-hand side).
Seed data: $\tau=(\frac{2}{3}-\frac{1}{3}\cos(2\phi))^a$ 
and $\sigma = \mu \overline \sigma$.
The solid blue line indicates a fold, whereas the red dashed line
indicate locations where blowup is inferred. 
The blue dotted lines at $\mu=0$ indicate 
one solution has zero volume and should be ignored.}
\end{figure}

\begin{figure}
  \includegraphics[height=\plotheight]{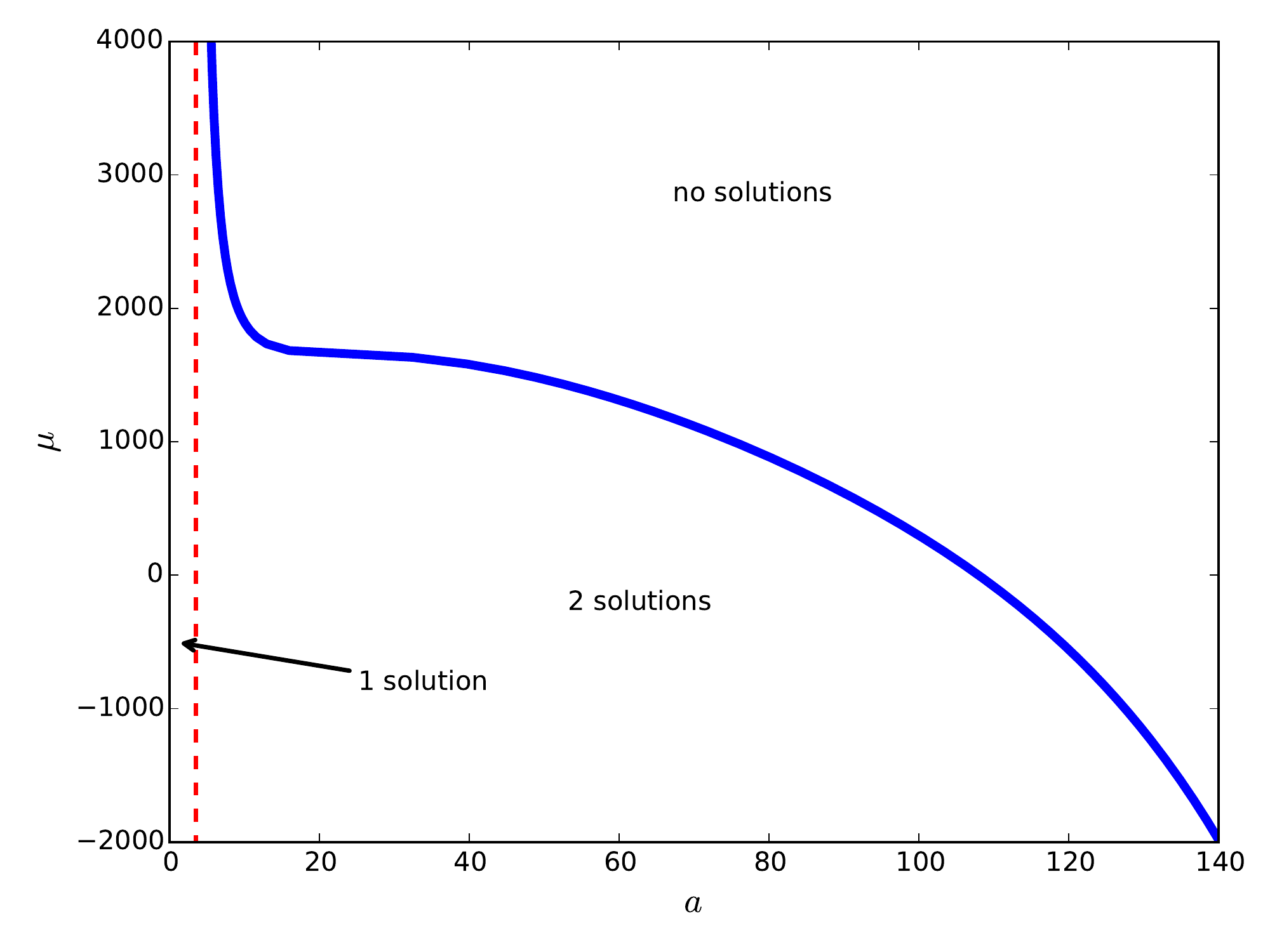}
\caption{  \label{fig:S2H2-NTC-solvability-neg}
Multiplicity of solutions found on an $S^2\times H^2$ 
with scalar curvatures $R=-0.5$
Seed data: $\tau=(\frac{2}{3}-\frac{1}{3}\cos(2\phi))^a$ 
and $\sigma = \mu \overline \sigma$.
The solid blue indicates a fold, whereas the red dashed line
indicate locations where blowup is inferred.}
\end{figure}

\section{Discussion}
   \label{sec:disc}

The limit equation \eqref{eq:limitEqn} appears to play a central
role in the solution theory of the conformal method, at least
for mean curvatures that do not change sign.  In the cases
where we could show that there is no solution of the limit equation
with the symmetry of the data
($S^1$-dependent seed data on $S^1\times M_2$), we were also unable to find
behavior that was any different from the near-CMC theory.  For the
remaining cases where we investigated constant-sign mean curvatures
$\tau=\xi^a$, there was a singular value $a_*$.  For $a>a*$
we found differences from the near-CMC theory: multiple solutions
or non-existence of solutions were the general rule, with
exceptions occuring only at transitions such as folds.  
As $a\to a_*$ from above,
we found solutions of the constraint equations with volumes that
appeared to approach infinity; in particular,
it was always possible to find such solutions with $\sigma\equiv 0$
We conjecture that at $a=a_*$ there
is a solution of the limit equation with $\alpha=1$, in which case
for each $a>a_*$ there is also a solution of the limit equation with 
$\alpha=a/a_*<1$.  That is, far-from-CMC behavior appears to occur
precisely when a solution to the limit equation exists 
for some $\alpha\in (0,1]$.

For sign changing mean curvatures, we found a corresponding transition to 
far-from-CMC phenomena once the mean curvature changed sign.  This was
certainly true for the Yamabe-negative and Yamabe-null data we examined,
and at least weakly so for the Yamabe-positive data of Figure
\ref{fig:S1-Ypos-mult} where narrow folds occurred, but unique solutions
were more typical.  In any event, in these examples, non-standard behaviour
only occurred  for 
conformal seed data where the mean curvature changed sign.  In may be a that
a better parameterization for sign-changing mean curvatures, different
from $\tau = 1 + a\xi$, would provide a sharper, more definitive transition.

Although Theorem \ref{thm:smallTT} does not apply, strictly speaking,
to our examples (as they possess nontrivial conformal Killing fields),
our observations were consistent with it.  Whenever we worked with
Yamabe-positive seed data we were able to find solutions of the
constraint equations, so long as $\sigma\not \equiv 0$ was small enough:
Figures \ref{fig:S1-Ypos-mult}, \ref{fig:LatDepSolvability}, \ref{fig:LatDep-NTC-Solvability}, \ref{fig:S2H2-Solvability-Pos} and \ref{fig:S2H2-NTC-solvability-nonneg} (right-hand side).  Conversely,
except for cases where we could show that there was not a solution
of the limit equation, far-from-CMC Yamabe-negative seed data
lead to non-existence sufficiently far into the non-CMC zone:
 Figures \ref{fig:S1-Yneg-folds}, \ref{fig:S2H2SolvabilityNonpos} (left-hand side) and \ref{fig:S2H2-NTC-solvability-neg}.  
This was true even at $\sigma\equiv 0$, in stark contrast to the
near-CMC theory. 

The situation for Yamabe-null data is harder to characterize.  Sometimes
it behaved like Yamabe-positive data 
(Figures \ref{fig:S1-Yzero-mult} and  \ref{fig:S2H2-NTC-solvability-nonneg}, left-hand side), with
solutions existing for sufficiently small TT tensors. Sometimes
it behaved like Yamabe negative data (Figure \ref{fig:S2H2SolvabilityNonpos}, right hand side), with solutions vanishing far enough into the far-from CMC zone.  Moreover, the analytical work of \cite{M11} shows that
other variations are also possible.

The conclusions of Theorem \ref{thm:ngn2} were found to hold generally,
even for mean curvatures that violate inequality \ref{eq:almostckf},
so long as there appeared to be a solution of the limit equation.
That is, for far-from-CMC Yamabe positive data (with constant-sign mean curvature), we found that there were at least two solutions 
when the TT tensor was small enough (and that there was at least one 
corresponding nonzero solution at $\sigma \equiv 0$).  

On the other hand, Theorem \ref{thm:Nguyen}, which 
also describes non-existence for Yamabe-positive far-from-CMC seed
data when the TT tensor is large, was not found to hold in general.
Indeed, we found a hodge-podge of apparent non-existence phenomena on Yamabe-positive seed data. Sometimes there was an immediate
onset of nonexistence behavior in the far-from-CMC zone (Figure \ref{fig:LatDepSolvability}).  Sometimes nonexistence was brought on by
scaling the TT by a large constant of one sign, but not for large
constants of the other sign (Figures \ref{fig:LatDep-NTC-Solvability} and \ref{fig:S2H2-NTC-solvability-nonneg} (right-hand side)). 
In one case (Figure \ref{fig:S2H2-Solvability-Pos}) there appeared
to be certain far-from-CMC seed data that did not lead to nonexistence
regardless of how large the TT tensor was scaled.  And in the 
sign-changing Yamabe-positive case we examined, existence 
appeared to be pervasive (Figure \ref{fig:S1-Ypos-mult}).
It is similarly hard to pin down precise non-existence behavior
for the Yamabe-null seed data we examined.  

We saw no apparent rule to describe the profiles of the various folds
we saw.  That is, we were unable to discern anything that might help
concretely predict the threshold of non-existence when scaling the TT
tensor or  the specific number of solutions corresponding to given seed data;
although zero, one or two solutions were typical, sometimes there
were more.  

\section{Conclusion}
   \label{sec:conc}

Our numerical work suggests that the
conformal method appears to suffer from pervasive drawbacks
as a parameterization of vacuum, non-CMC solutions of the 
constraint equations. At least among the data we considered,
the general rule was multiple solutions or no solutions at all
once the conformal seed data was sufficiently far-from-CMC.
Because of the limitations of AUTO, we were only able to 
examine highly symmetric seed data, and we therefore only probed
a select few, very special examples.  Nevertheless, it is difficult
to imagine that the many cases of multiple solutions we
found are not stable under small perturbations of the metric that
violate symmetry.

Our results suggest a couple of theorems that might be reasonable
targets for future efforts.  For example, can non-existence of solutions
for sufficiently far-from-CMC Yamabe-negative seed data be established?
However, we caution that it is possible that 
describing the details of the conformal
method for far-from-CMC data will lead to a fuller understanding 
of the conformal method, but also to 
nothing useful about general relativity.  Unless
there is some physics associated with to the multiplicity of solutions or
the various shapes of the folds, it may simply be that the conformal
method is an excellent parameterization of the CMC solutions
of the constraints that breaks down as a chart on the larger
constraint manifold.  Since the conformal method is the only 
general tool available for constructing solutions of constraint
equations \textit{de novo}, it raises the question of whether
a suitable alternative parameterization for non-CMC initial data
exists. One potential was proposed in \cite{Maxwell:2014c} and
examined for near-CMC constructions in \cite{HMM2017},
but its properties in the far-from-CMC setting are unknown
and the broader question of finding a well-behaved
global parameterization of solutions of the constraint equations is
essentally open.
   
\section{Acknowledgments}
   \label{sec:ack}

JD was supported in part by NSF DMS/RTG Award 1345013 and NSF DMS/FRG Award 1262982.
MH was supported in part by NSF DMS/FRG Award 1262982 and NSF DMS/CM Award 1620366.
TK was supported in part by NSF DMS/RTG Award 1345013.
DM was supported in part by NSF DMS/FRG Award 1263544.

\appendix
\section{A constant norm TT tensor on $S^1\times S^2$}
\label{apdx:TT}
In this section we construct a transverse-traceless tensor
on $S^1\times S^2$ that has constant norm and is pointwise
orthogonal to $\ck W$ when $W$ is an $S^1$-dependent
vector field pointing along $S^1$.

Consider normal (polar) coordinates $(r,\theta)$
on the unit round sphere $S^2$ centered at the north pole, so that 
the metric has the form $g=dr^2+\sin^2(r)d\theta^2$.  Let
$\omega = \sin(r) d\theta$.  It is clear that $\omega$ fails
to be continuous at the north and south poles of $S^2$, but is otherwise smooth.  A straightforward
computation shows that $|\nabla \omega|=\cot(r)$.  The singularity at $r=0$ is $O(r^{-1})$, with similar remarks holding at $r=\pi$. Hence
$\omega\in W^{1,p}(S^2)$ for any $p<2$.  Moreover, $\div \omega = 0$ in the region where $\omega$ is smooth (i.e. almost
everywhere) and therefore $\omega$ is weakly divergence free.

Now let $s$ denote a unit speed parameter on $S^1$ and set
\[
\hat \sigma=\omega\otimes ds + ds \otimes \omega
\]
on $S^1\times S^2$.  Clearly $\sigma$ is trace-free. Moreover
since $\omega$ and $ds$ are both divergence free, $\sigma$ is a transverse-traceless tensor on $S^1\times S^2$.  Finally, 
$$
|\sigma|^2 = 2\abs{\omega}^2\abs{dz}^2+4\left<\omega,dz\right> = 2\abs{\omega}^2 = 2 \sin^2(r) \abs{d\theta}^2 =  2 \sin^2(r) (\sin(r))^{-2} =2,
$$
except at $r=0$ and $r=\pi$.  That is, $\abs{\sigma}^2 = 2$ almost everywhere.  Hence $\hat\sigma$ is a constant-norm $W^{1,p}$ transverse-traceless tensor on $S^1\times S^2$ for any $p<2$.
Although this level of regularity is not ideal, it falls
within the a category of regularity  
easily handled for the conformal method (e.g. \cite{CB04}).

If $W=w(s)\partial_s$ then $\ck W = 2w'(\frac{2}{3} ds\otimes ds-\frac{1}{3}g^\circ$), where $g^\circ$ is the round metric on the sphere.  Since
$\hat\sigma$ only has $ds\otimes d\theta$ and $d\theta\otimes ds$ components,
it is pointwise orthogonal to any such $\ck W $.

\bibliographystyle{abbrv}
\bibliography{bib/maxwell-gathers-thoughts,bib/mjh-gr,bib/mjh-papers,bib/mjh-more,bib/mjh-moreyet,bib/books,bib/paper1}
\end{document}